\documentclass[12pt]{iopart}

\pdfoutput=1 

\expandafter\let\csname equation*\endcsname\relax
\expandafter\let\csname endequation*\endcsname\relax

\usepackage{amsmath,amssymb,amsfonts,amsthm}

\usepackage{graphicx}
\usepackage{enumitem}
\usepackage{bm}
\usepackage{rotating}
\usepackage{hyperref}
\usepackage{pbox}
\usepackage{array}
\usepackage{mathtools}
\usepackage{leftidx}
\usepackage{lmodern}
\usepackage[normalem]{ulem}
\usepackage{braket}
\usepackage{tensor}
\usepackage[usenames,dvipsnames]{xcolor}
\usepackage{float}
\usepackage{setspace}
\usepackage{CJKutf8}
\usepackage{dsfont}
\usepackage{mathrsfs}
\usepackage[many]{tcolorbox}
\usepackage{dsfont}
\usepackage[numbers, sort&compress]{natbib}

\newcommand{\Z}{\mathbb{Z}}
\newcommand{\R}{\mathbb{R}}
\newcommand{\A}{\mathcal{A}}
\newcommand{\C}{\mathbb{C}}
\newcommand{\M}{\mathcal{M}}
\newcommand{\W}{\mathcal{W}}
\newcommand{\CS}{{C^\infty_0(\M)}}
\newcommand{\dd}{\mathrm{d}}
\newcommand{\rr}[1]{\left(#1\right)}
\newcommand{\bx}{{\bm{x}}}

\newcommand{\bk}{{\bm{k}}}
\newcommand{\sx}{\mathsf{x}}
\newcommand{\sy}{\mathsf{y}}
\newcommand{\sz}{\mathsf{z}}

\newcommand{\ii}{\mathsf{i}}
\DeclareMathOperator{\supp}{\text{supp}}
\renewcommand{\tr}{\mathrm{Tr}}
\renewcommand{\tilde}{\widetilde}
\renewcommand{\bar}{\overline}

\newcommand{\Sol}{\mathsf{Sol}}

\renewcommand{\Re}{\mathrm{Re}}
\renewcommand{\Im}{\mathrm{Im}}
\newcommand{\Wk}{{\tilde{\mathcal{W}}(\R^n)}}

\newcommand{\ketbra}[2]{{\left| {#1} \right\rangle \!\!\left\langle {#2} \right|}}
\newcommand{\kg}{{\textsc{KG}}}
\newcommand{\fock}{{\mathfrak{F}(\mathcal{H})}}
\DeclareMathOperator{\Id}{\text{Id}}

\newtheorem{proposition}{Proposition}
\newtheorem{assumption}{Assumption}
\newtheorem{definition}{Definition}

\newtheorem{lemma}{Lemma}

\newcommand{\openone}{\mathds{1}}

\makeatletter
\long\def\@makefntext#1{\parindent 1em\noindent 
 \makebox[1em][l]{\footnotesize\rm$\m@th^{\arabic{footnote}}$}%
 \footnotesize\rm #1}
\def\@makefnmark{\hbox{$^{\arabic{footnote}}\m@th$}}
\def\@thefnmark{\arabic{footnote}}
\makeatother

\begin{document} 

\title{The Unruh-DeWitt model and its joint interacting Hilbert space}

\author{Erickson Tjoa\footnote{Corresponding author.}}
\ead{erickson.tjoa@mpq.mpg.de}
\address{Max-Planck-Institut f\"ur Quantenoptik, Hans-Kopfermann-Stra\ss e 1, D-85748 Garching, Germany}

\author{Finnian Gray}
\ead{finnian.gray@univie.ac.at}
\address{University of Vienna, Faculty of Physics, Boltzmanngasse 5, A 1090, Vienna}

\date{\today}

\begin{abstract}
   In this work we make the connection between the Unruh-DeWitt particle detector model applied to quantum field theory in curved spacetimes and the rigorous construction of the spin-boson model. With some modifications, we show that existing results about the existence of a spin-boson ground state can be adapted to the Unruh-DeWitt model. In the most relevant scenario involving massless scalar fields in (3+1)-dimensional globally hyperbolic spacetimes, where the Unruh-DeWitt model describes a simplified model of light-matter interaction, we argue that common choices of the spacetime smearing functions regulate the ultraviolet behaviour of the model but can still exhibit infrared divergences. In particular, this implies the well-known expectation that the joint interacting Hilbert space of the model cannot be described by the tensor product of a two-dimensional complex Hilbert space and the Fock space of the vacuum representation. We discuss the conditions under which this problem does not arise and the relevance of the operator-algebraic approach for better understanding of particle detector models and their applications. Our work clarifies the connection between obstructions due to Haag’s theorem and infrared bosons in the spin-boson models, and paves the way for more rigorous study of entanglement and communication in the UDW framework involving multiple detectors.
\end{abstract}

\maketitle
\flushbottom
\setcounter{footnote}{0}


\section{Introduction}

Broadly speaking, relativistic quantum information (RQI) is concerned with the role of relativity in quantum information-processing tasks, or the use of quantum information theory for better understanding of fundamental theories such as relativistic quantum field theories (QFT). Even if we exclude approaches to quantum gravity, investigations in RQI vary vastly in methodologies and objectives. A highly non-exhaustive list includes studying relativistically consistent measurement theories \cite{fewster2020measurement,bostelmann2021impossible,polo2021detectorbased,jubb2022causal,pranzini2023detector}; operationalizing field-theoretic phenomena such as Unruh and Hawking effects using localized quantum-mechanical probes \cite{hawking1975particle,wald1994quantum,Crispino2008review,Unruh1979evaporation,DeWitt1979,tjoa2022unruh}; rigorous study of the entanglement structure and complexity of quantum field theories \cite{hollands2017entanglement,hollands2023channel,sanders2023separable,summers1985bell,summers1987bell,casini2020entanglement,casini2021entropic,longo2018relative}; understanding the relationship between information-theoretic causality and relativistic causality \cite{vilasini2022embedding,Hardy2007towards,zych2019bell,Costa2016causal,paunkovic2020causal}; constructing relativistic protocols \cite{reznik2003entanglement,reznik2005violating,Valentini1991nonlocalcorr,pozas2015harvesting,pozas2016entanglement,Jonsson2014cavityQED,tjoa2021harvesting,Simidzija2020capacity,Landulfo2016communication,tjoa2022teleport,lapponi2023relativistic,kent1999bitcommit,lo1998quantum,adlam2015crypto,buhrman2014position,vilasini2019composable}; and many others. 

The use of localized quantum-mechanical probes to study relativistic QFT is by now well-known, with the simplest example being the so-called Unruh-DeWitt (UDW) particle detector model \cite{Unruh1979evaporation,DeWitt1979}. This model is based on coupling a two-level system (``qubit detector'') to the field locally in spacetime, and it has been generalized in many ways \cite{Lopp2021deloc,Tales2020GRQO,Bruno2020time-ordering,Tjoa2020vaidya,Aubry2014derivative,perche2022localized,Nadine2021delocharvesting,doukas2013unruh,hu2022qhorad,hotta2020duality,perche2023particle,gale2023relativisticCOM,tjoa2023qudit}. For practical reasons, the calculations are mostly made based on the interaction picture to isolate the effect of the interaction in the Hamiltonian (see, e.g., \cite{Jonsson2024chainmappingmethods,hu2007exactQHO}, for some exceptions), and in certain restricted cases the unitary associated with the interacting piece can be computed exactly \cite{Landulfo2016communication,tjoa2023nonperturbative,lapponi2023relativistic,tjoa2022teleport,tjoa2022fermi,tjoa2022capacity,Simidzija2018no-go,Simidzija2020capacity}. Note that even in this latter case we cannot claim that we are ``exactly'' solving the dynamics: this requires us to either diagonalize the \textit{full} Hamiltonian of the UDW model, or at the very least show that the joint system has a ground state. Both tasks are clearly non-trivial given that even the {(simpler)} quantum Rabi model and some of its generalizations were only exactly solved recently \cite{braak2011rabi,xie2017quantum}.

In this work we revisit the UDW model from a different direction to re-assess the validity of the model. The reason is closely related to {Haag's theorem} \cite{earman2006haag}, or more generally the fact that there are many unitarily inequivalent representations {of} systems with infinitely many degrees of freedom. Under very generic situations,  one can show that interaction picture is of limited applicability in QFT with interactions and the UDW model is not expected to be exempt from this. Since the interaction picture forms the basis of virtually all known calculations in the literature on the UDW model, knowing its limitation is of fundamental importance especially since in some cases one would like to make “non-perturbative” claims, such as when using gapless detector or delta-coupled detectors \cite{Landulfo2016communication,tjoa2023nonperturbative,lapponi2023relativistic,tjoa2022teleport,tjoa2022fermi,tjoa2022capacity,Simidzija2018no-go,Simidzija2020capacity}. For example, it is known that for relativistic communication protocol through a quantum field, it is essential that the interaction involves strong coupling beyond perturbation theory for the channel capacity to be non-trivial \cite{tjoa2022capacity,Landulfo2016communication,Simidzija2020capacity}.

In more precise terms, our task is to compare the UDW model with the known rigorous formulation of the spin-boson model \cite{fannes1988equilibrium,spohn1989spinboson,amann1991spinboson,hasler2011ground,hasler2021existence,DeRoeck2015SBscatter,deBievre2006unruh} and see how obstructions similar to Haag’s theorem can arise in the UDW framework. Fortunately, we are greatly aided by extensive literature on spin-boson Hamiltonian \cite{fannes1988equilibrium,spohn1989spinboson,amann1991spinboson,hasler2011ground,hasler2021existence,DeRoeck2015SBscatter,deBievre2006unruh}, van Hove Hamiltonian {\cite{VanHove:1952jxs,derezinski2003vanHove,fewster2019algebraic}}, and more generally the Pauli-Fierz Hamiltonian \cite{derezinski2004scattering} and hence we do not need to reinvent the wheel. Inspired by the van Hove model describing a scalar field sourced by a classical time-independent current (see  {\cite{VanHove:1952jxs,derezinski2003vanHove,fewster2019algebraic}}), we first identify that Haag’s theorem applied to the “infrared (IR) regime” of the van Hove model and infinite many soft boson productions in the spin-boson model are essentially the same problem. Hence, by lifting this to the UDW model, we show that the UDW model suffers from the same obstruction in certain regimes where certain integrals quantifying the particle production is IR-divergent. 

Two interesting applications of this result  {are} that (1) massless fields in curved spacetimes can avoid IR divergences of this type, as spacetime curvature can modify the dispersion relation (by giving effective mass to the field modes), and (2) commonly used Gaussian smearing functions cannot in general remove the IR divergences, although Haag’s obstruction in the “ultraviolet (UV) regime” is always avoided.
In order to lift the spin-boson result to the UDW model, we need to modify the model such that the joint UDW Hamiltonian is time-independent, i.e., without a compactly supported switching function. We argue that this is advantageous for two reasons. First, we would like to be able to view the UDW model as a closed system, the same way that physical models in quantum optics such as quantum Rabi model or Jaynes-Cummings model are used. Second, for technical reasons this allows us to rigorously construct the UDW model, its Hilbert space representations and its dynamics using operator-algebraic techniques, rather than working with the model via the scattering theory formalism (e.g., Haag-Ruelle type \cite{derezinski2004scattering,DeRoeck2015SBscatter,dybalski2019scattering}) that is much harder and less natural for the typical problems of interest in the UDW framework. 

Our modification of the UDW model comes with an apparent price, namely that we lose the apparent temporal localization property of the UDW model. In the standard UDW model, it is possible to have two particle detectors to have their interactions be fully spacelike separated because the interaction can be appropriately switched off. The time-independent formulation does not have this property, and we argue that this is not a problem precisely because in operator-algebraic language we are demanding that observables at spacelike separation commute and not at the level of the joint Hamiltonian. Indeed, the time-independent formulation of the UDW model is closer in spirit to the way we view many closed interacting systems in the Standard Model where interactions are “always on” at the level of the Hamiltonian and Lagrangian, and relativistic causality is enforced, for example, through the cluster decomposition property. This will be physically relevant when studying multiple detector settings, and our work is a first step before going towards the multipartite generalization of the modified UDW model.


This paper is organized as follows. In Section~\ref{sec: setup} we briefly review the algebraic framework for quantization of scalar field theory in curved spacetimes and discuss the possible of ultraviolet and infrared divergences when a classical current source is present (the so-called \textit{van Hove model}  {\cite{VanHove:1952jxs,fewster2019algebraic,derezinski2003vanHove}}). In Section~\ref{sec: UDW-model} we review the standard construction of the UDW model and compare it with the rigorous formulation of the spin-boson model. In Section~\ref{sec: consequences} we revisit the UDW model and assess the validity of the model based on the analysis of the spin-boson model. Finally, in Section~\ref{sec: discussion} we provide some discussions on the implications of our analysis. We adopt the mostly-plus signature for the metric and use the natural units $c=\hbar=1$.

\section{Scalar field theory in curved spacetimes}
\label{sec: setup}

Let us first review the construction of scalar field theory in curved spacetimes (see, e.g., \cite{birrell1984quantum,wald1994quantum,fewster2019algebraic,Khavkhine2015AQFT} for detailed reviews). 
Let $\M$ be an $(n+1)$-dimensional globally hyperbolic spacetime with metric $g_{ab}$. A Klein-Gordon scalar field $\phi:\M\to \R$ satisfies the curved spacetime generalization of the wave equation, known as the Klein-Gordon equation, given by
\begin{align}
    (\nabla_a\nabla^a-m^2-\xi R)\phi = 0\,,
    \label{eq: KGE}
\end{align}
where $\nabla_a\nabla^a = g^{ab}\nabla_a\nabla_b$, $\nabla_a$ is the covariant derivative with respect to the Levi-Civita connection, $R$ is the Ricci scalar curvature. The parameter $\xi\geq 0$ prescribes (non-)minimal coupling to the scalar curvature and $m$ is the mass parameter. Eq.~\eqref{eq: KGE} generalizes the standard wave equation $(\partial_a\partial^a-m^2)\phi=0$ in flat spacetime, where $\partial_a\partial^a=\eta^{ab}\partial_a\partial_b$ and $\eta_{ab}$ is the flat Minkowski metric. Global hyperbolicity guarantees that the spacetime admits a foliation  $\M\cong \R\times \Sigma$ where $\Sigma$ is a Cauchy surface and there is a good notion of global time ordering, i.e., we can speak of ``constant-time slices''.

We start by considering quantization of the scalar field in the language of algebraic quantum field theory (AQFT) \cite{fewster2019algebraic,hollands2017entanglement,wald1994quantum,Khavkhine2015AQFT}. In the algebraic framework,  we first construct the \textit{algebra of observables} $\A(\M)$ for the field theory as well as quantum states $\omega$ on which $\A(\M)$ acts. The canonical quantization procedure is then recovered through a representation of the canonical commutation relations (CCR) associated to the pair $(\A(\M),\omega)$ via the so-called \textit{Gelfand-Naimark-Segal} (GNS) \textit{construction}. Finally, one then restricts the choice of algebraic states to a physically reasonable subclass called\textit{ Hadamard states }that have the correct short-distance singularities and for which the quantum stress-tensor has well-defined vacuum expectation values \cite{KayWald1991theorems,Radzikowski1996microlocal}.

Let $f\in \CS$ be a smooth compactly supported test function on $\M$ and let  $E^{\text{R/A}}\equiv E^{\text{R/A}}(\sx,\sy)$ be the {\textit{retarded/advanced propagators} associated with the Klein-Gordon operator $\hat P$ that satisfies $\hat P(E^{\text{R/A}} f) = f$}, where
\begin{align}
    E^{\text{R/A}} f\equiv (E^{\text{R/A}} f)(\sx) \coloneqq \int \dd V'\, E^{\text{R/A}} (\sx,\sx')f(\sx') \,,
\end{align}
and $\dd V' = \dd^n\sx'\sqrt{-g}$ is the proper volume element. The \textit{causal propagator} is defined to be the {advanced-minus-retarded propagator\footnote{The conventions used in the literature are not completely standardized, in addition to the different choice of metric signature. See \ref{appendix: green-function} for more details. } $E=E^\text{A}-E^\text{R}$}. If $O$ is an open neighbourhood of some Cauchy surface $\Sigma$ and $\varphi \in \Sol_\R(\M)$ is any real solution to Eq.~\eqref{eq: KGE} with compact Cauchy data, then it is known that there exists $f\in \CS$ with $\supp(f)\subset O$ such that $\varphi=Ef$ \cite{Khavkhine2015AQFT}. In other words, the compactly supported test functions serve as labels for the space of solutions for the Klein-Gordon equation.

\subsection{Algebra of observables}

The quantum theory for the scalar field $\phi$ is defined by specifying its \textit{algebra of observables}  $\A(\M)$. The quantization of $\phi$ is viewed as an $\R$-linear mapping {from the space of test functions into the algebra,} 
\begin{align}
    \hat\phi: C^\infty_0(\mathcal{M})&\to \A(\M)\,,\quad f\mapsto \hat\phi(f)
\end{align}
where the \textit{smeared field operator} is given by
\begin{align}
    \hat{\phi}(f)\coloneqq \int\dd V\,\hat{\phi}(\sx)f(\sx)\,.
\end{align}
That is, $\A(\M)$ is a free algebra generated by $\hat{\phi}(f)$, its formal adjoint $\hat{\phi}(f)^\dagger$, and the identity operator $\openone$ where $f\in \CS$, together with the relations:
\begin{enumerate}[leftmargin=*,label=(\alph*)]
    \item $\hat\phi(f)^\dag = \hat\phi(f)$ for all $f\in \CS$;
    \item $\hat\phi(\hat Pf) = 0$ for all $f\in \CS$, where $\hat{P}=\nabla_a\nabla^a-m^2-\xi R$ is the Klein-Gordon wave operator.
    \item $[\hat\phi(f),\hat\phi(g)] = \ii E(f,g)\openone $ for all $f,g\in \CS$, where $E(f,g)$ is the \textit{smeared} causal propagator
    \begin{align}
        E(f,g)\coloneqq \int \dd V f(\sx) (Eg)(\sx)\,.
    \end{align}
\end{enumerate}

Therefore $\A(\M)$ forms a unital $*$-algebra. The usual unsmeared field operator $\hat{\phi}(\sx)$ is to be regarded as an operator-valued distributions. By the time-slice axiom, it suffices that the algebra is generated by smeared field operators where $\supp(f)\subset O$ for some open subset $O\subset \M$ that contains a Cauchy slice $\Sigma$, though for non-interacting scalar fields this is comes for free \cite{benini2015models}.

For our purposes we may also work with the ``exponentiated version'' of $\hat\phi(f)$ that generates a \textit{Weyl algebra} $\W(\M)$. This is a unital $C^*$-algebra generated by elements that formally take the form 
\begin{align}
    W(Ef) \equiv 
    {e^{\ii\hat\phi(f)}}\,,\quad f\in \CS\,,    \label{eq: Weyl-generator}
\end{align}
obeying \textit{Weyl relations}:
\begin{equation}
    \begin{aligned}
    W(Ef)^\dagger &= W(-Ef)\,,\\
    W(E (Pf) ) &= \openone\,,\\
    W(Ef)W(Eg) &= e^{-\frac{\ii}{2}E(f,g)} W(E(f+g))
    \end{aligned}
    \label{eq: Weyl-relations}
\end{equation}
where $f,g\in \CS$. The Weyl algebra $\W(\M)$ is more convenient as it admits a $C^*$-norm so that it can be promoted to a unique (up to isomorphism) unital $C^*$-algebra \cite{bratteli2013operator}. We can thus realize elements of the Weyl algebras as bounded operators acting on a Hilbert space via the GNS construction \cite{wald1994quantum,Khavkhine2015AQFT,fewster2019algebraic} that relates the algebraic framework with the standard canonical quantization.

\subsection{Algebraic states and quasifree representation}

An \textit{algebraic state} is defined to be a $\C$-linear functional $\omega:\W(\M)\to \C$ (similarly for $\A(\M)$) such that 
\begin{align}
    \omega(\openone) = 1\,,\quad  \omega(A^\dagger A)\geq 0\quad \forall A\in \W(\M)\,.
    \label{eq: algebraic-state}
\end{align}
In words, the algebraic state is essentially a map from observables to expectation values. We say that the state $\omega$ is pure if it cannot be written as $\omega= \alpha \omega_1 + (1-\alpha)\omega_2$ for any $\alpha\in (0,1)$ and any two algebraic states $\omega_1,\omega_2$ and it is mixed otherwise. 

The basic set of physical quantities in the scalar QFT is given by its \textit{$n$-point correlation function} (also known as the Wightman $n$-point functions) of the field operators
\footnote{ {The notation for the Wightman function $\mathsf{W}(f_1,...,f_n)$ should neither be confused with the italicized $W(Ef)$ describing the Weyl generator nor with the calligraphic $\cal W(M)$ representing the whole algebra.}}
\begin{align}
    \mathsf{W}(f_1,...,f_n)\coloneqq \omega(\hat\phi(f_1)...\hat\phi(f_n))\,,\quad f_j\in \CS\,.
    \label{eq: n-point-functions}
\end{align}
Although the RHS seems to involve the algebra of unbounded operators $\A(\M)$, the GNS representation of the Weyl algebra $\W(\M)$ lets us calculate these correlation functions by taking derivatives: for example, the smeared Wightman two-point function is formally computed as
\begin{align}\label{eq: Wightman-formal-bulk}
    &\mathsf{W}(f,g) \equiv -\frac{\partial^2}{\partial s\partial t}\Bigg|_{s,t=0}\!\!\!\!\!\!\!\!\omega(e^{\ii\hat\phi(sf)}e^{\ii\hat\phi(tg)})\,.
\end{align}
The RHS is only formal as the Weyl algebra itself does not have the right topology to allow both exponentiation and taking derivatives. The RHS must therefore be understood in a sufficiently regular GNS representations with respect to a state $\omega$ \cite{fewster2019algebraic,Khavkhine2015AQFT,derezinski2006introduction}.

Not every algebraic state is physically relevant and it is generally accepted that any physically reasonable state for the scalar QFT should be a \textit{Hadamard state} \cite{KayWald1991theorems,Radzikowski1996microlocal}. These are the states whose renormalized stress-energy tensor (and hence the Hamiltonian) is well-defined and Hadamard states have the correct short-distance singularities. Furthermore, there is a subfamily of states known as \textit{quasifree states} where explicit constructions can be given. This family has three nice properties: (i) the GNS representation (called \textit{quasifree representations}) will correspond to the usual \textit{Fock representation} in the canonical quantization approach, (ii) the state is completely specified by its two-point correlation functions; (iii) the quasifree representation is sufficiently well-behaved that if $\pi_\omega$ is a quasifree representation for $\A(\M)$ and $\Pi_\omega$ is a quasifree representation for $\W(\M)$, then
\begin{align}
    \pi_\omega(\hat{\phi}(f)) = -\ii\frac{\dd}{\dd t}\Pi_\omega(W(tEf))\Bigr|_{t=0}\,.
\end{align}
In other words, in the quasifree representation we can indeed view $\Pi_\omega(W(Ef))=e^{\ii\pi_\omega(\hat{\phi}(f))}$, which justifies the shorthand in Eq.~\eqref{eq: Weyl-generator} (see \cite{ruep2021weakly} for nice explanation of this point). Indeed, a representation where $\hat{\phi}(f)$ makes sense is called a \textit{regular representation}\footnote{{In a regular representation, the smeared field operator $\pi_\omega(\hat{\phi}(f))$ is essentially self-adjoint over a dense invariant linear subspace $\mathcal{D}_\omega$ of the GNS Hilbert space $\mathcal{H}_\omega$ for all $f\in\CS$ \cite{Khavkhine2015AQFT}.}} \cite{derezinski2006introduction,Khavkhine2015AQFT}.

To see how the Fock representation arises, we first consider how the definition of quasifree states is used in practice. First, we note that the solution space $\Sol_\R(\M)$ is a real symplectic vector space equipped with a symplectic form $\sigma:\Sol_\R(\M)\times\Sol_\R(\M)\to \R$ 
\begin{align}
    \sigma(\varphi_1,\varphi_2) \coloneqq \int_{\Sigma_t}\!\! {\dd\Sigma^a}\,\Bigr[\varphi_{{1}}\nabla_a\varphi_{{2}} - \varphi_{{2}}\nabla_a\varphi_{{1}}\Bigr]\,,
    \label{eq: symplectic form}
\end{align}
where $\dd \Sigma^a = -{n}^a \dd\Sigma$, $-{n}^a$ is the inward-directed unit normal {(i.e. past directed)} to the Cauchy surface $\Sigma_t$, and $\dd\Sigma = \sqrt{h}\,\dd^3\bx$ is the induced volume form on $\Sigma_t$ \cite{poisson2009toolkit,wald2010general}. 

Any quasifree state $\omega_\mu$ is associated with a \textit{real-bilinear inner product} $\mu:\Sol_\R(\M)\times\Sol_\R(\M)\to \R$ that satisfies \cite{KayWald1991theorems}
\begin{align}
    |\sigma(Ef,Eg)|^2 \leq \mu(Ef,Ef)\mu(Eg,Eg)\,,
\end{align}
where we recall that $Ef,Eg\in \Sol_\R(\M)$ for all $f,g\in\CS$. The inequality is saturated if $\omega_\mu$ is a pure state. A quasifree state is then \textit{defined} to be those that satisfy
\begin{align}
    \omega_\mu(W(Ef)) \coloneqq e^{-\mu(Ef,Ef)/2}\,.
    \label{eq: quasifree-def}
\end{align}
Due to the form of \eqref{eq: quasifree-def}, in modern context quasifree states are sometimes called \textit{Gaussian states} \cite{Khavkhine2015AQFT}. In what follows we simply write $\omega$ instead of $\omega_\mu$ and the context should make it clear.

The definition of quasifree states is not useful unless we can compute the induced norm $||Ef|| \coloneqq \sqrt{\mu(Ef,Ef)}$. To achieve this, we first extend $\Sol_\R(\M)$ into the space of \textit{complex} solutions $\Sol_\C(\M)$ of \eqref{eq: KGE} and define the \textit{Klein-Gordon (KG) bilinear product} by 
\begin{align}
    \braket{\varphi_1,\varphi_2}_{\textsf{KG}}\coloneqq \ii\sigma(\varphi_1^*,\varphi_2)
\end{align}
for $\varphi_1,\varphi_2\in \Sol_\C(\M)$. It was shown in \cite{KayWald1991theorems} that for a quasifree state associated with the bilinear map $\mu$, there exists a subspace $\mathcal{H}\subset\Sol_\C(\M)$ such that $(\mathcal{H},\braket{\cdot,\cdot}_{\textsf{KG}})$ is a Hilbert space --- known as the  \textit{one-particle Hilbert space} --- and an $\R$-linear map $K:\Sol_\R(\M)\to\mathcal{H}$ such that for all $\varphi_1,\varphi_2\in\Sol_\R(\M)$ \cite{KayWald1991theorems}
\begin{enumerate}[label=(\alph*)]
    \item $K\Sol_\R(\M)+\ii K\Sol_\R(\M)$ is dense in $\mathcal{H}$;
    \item $\mu(\varphi_1,\varphi_2) = \Re\braket{K\varphi_1,K\varphi_2}_{\textsf{KG}}$;
    \item $\sigma(\varphi_1,\varphi_2) = 2\Im\braket{K\varphi_1,K\varphi_2}_{\textsf{KG}}$;
    \item $\Sol_\C(\M)\cong \mathcal{H}\oplus\mathcal{\overline H}$, where $\mathcal{\overline{H}}$ is the complex-conjugate Hilbert space of $\mathcal{H}$ abd $\braket{u,v}_{\textsf{KG}}=0$ for all $u\in\mathcal{H}$ and $v\in\mathcal{\overline H}$.
\end{enumerate}
In the language of canonical quantization, the map $K$ retains the ``positive-frequency part'' of a real solution to the Klein-Gordon equation. The specification of $(\mathcal{H},K)$ is known as the \textit{one-particle structure} \cite{KayWald1991theorems}.

The Wightman two-point function is given by \cite{KayWald1991theorems}
\begin{align}
    \mathsf{W}(f,g) &= \braket{KEf,KEg}_{\textsf{KG}} = \mu(Ef,Eg) + \frac{\ii}{2}E(f,g)\,,
\end{align}
where we have used the fact that $\sigma(Ef,Eg) = E(f,g)$. Since $E(f,g)$ is antisymmetric, it follows that $||Ef||^2 = \mu(Ef,Ef) \equiv \mathsf{W}(f,f)$. Hence the quasifree condition boils down to computing symmetrically smeared Wightman function $\mathsf{W}(f,f)$. The calculations of $\mathsf{W}(f,g)$ can be done in any sufficiently regular representations via the GNS construction.

\subsection{Relationship with canonical quantization}

The Fock representation that arises from the GNS representation of $\A(\M)$ can be summarized as follows (see \cite{fewster2019algebraic,wald1994quantum,Khavkhine2015AQFT,bratteli2002operatorv2} for more details). The GNS Hilbert space $\mathcal{H}_\omega$ for the quasifree representation is given by the Fock space over the one-particle Hilbert space $\mathcal{H}$
\begin{align}
    \mathcal{H}_\omega \equiv \fock = \bigoplus_{n=0}^\infty \mathcal{H}^{\odot n}\,,
\end{align}
where $\mathcal{H}^{\odot 0}\cong \C$ and $\mathcal{H}^{\odot n}$ means symmetrized direct sum (the $n$-particle sector of the Fock space). In the Fock representation, the field operator can be written as
\begin{align}
    \hat{\phi}(f) \equiv \pi_\omega(\hat\phi(f)) = \hat{a}(({KEf})^*) + \hat{a}^\dagger(KEf)\,,
\end{align}
where we drop $\pi_\omega$ when it is clear that we are in the Fock representation. Note that if we consider complex smearing functions $f:\M\to \C$, we can write
\begin{align}
    \hat{\phi}(f) \equiv \hat{\phi}(\Re f) + \ii \hat{\phi}(\Im f)\,.
\end{align}
The operators $\hat{a}(u^*),\hat{a}^\dagger(v)$ are \textit{smeared ladder operators}\footnote{Note that there are various conventions on how to label the smeared ladder operators (see \cite{tjoa2022capacity}). The convention we pick here agrees with \cite{wald1994quantum} in that it maintains $\hat{a},\hat{a}^\dagger$ as being linear in their arguments. The convention in \cite{fewster2019algebraic} that does not include complex conjugation in the argument for $\hat{a}(\cdot)$ would make $\hat{a}$ an anti-linear map while $\hat{a}^\dagger(\cdot)$ is linear.} obeying the CCR
\begin{align}
    [\hat{a}(u^*),\hat{a}^\dagger(v)] = \braket{u,v}_{\textsf{KG}}\openone
\end{align}
on a suitable dense domain of the Fock space (since these operators are unbounded).

The standard canonical approach is recovered by working with the unsmeared field operator $\hat{\phi}(\sx)$ and considering the Fourier mode decomposition
\begin{align}
    \hat{\phi}(\sx) = \int\dd^n\bk\,\hat{a}_\bk u_\bk(\sx) + \hat{a}_\bk^\dagger u_\bk^*(\sx)\,,
\end{align}
where $u_\bk(\sx)$ and $u_\bk^*(\sx)$ are positive- and negative-frequency modes normalized to Dirac delta functions via the Klein-Gordon product:
\begin{subequations}
\begin{align}
    \braket{u_\bk,u_{\bk'}}_{\textsf{KG}} &= \delta^n(\bk-\bk')\,,\\
    \braket{u^*_\bk,u^*_{\bk'}}_{\textsf{KG}} &= -\delta^n(\bk-\bk')\,,\\
    \braket{u_\bk,u^*_{\bk'}}_{\textsf{KG}} &= 0\,.
\end{align}
\end{subequations}
These modes are not proper elements of the one-particle Hilbert space but they are convenient to work with. In particular, we see that $\hat{a}^{\phantom{*}}_\bk\equiv \hat{a}(u_\bk^*)$ and $\hat{a}_\bk^\dagger\equiv \hat{a}^\dagger(u_\bk)$. Using these, we can decompose the solution $Ef$ in the ``eigenmode basis'' $\{u_\bk,u_\bk^*\}$
\begin{align}
    Ef \equiv \int\dd^n\bk \,\braket{u_\bk,Ef}_{\textsf{KG}}u_\bk + \braket{u^*_\bk,Ef}_{\textsf{KG}}u_\bk^*\,.
\end{align}
Using the fact that
\begin{align}
   {\ii}E(\sx,\sx') &= \int\dd^n\bk\,u_\bk(\sx)u_\bk^*(\sx') - u_\bk^*(\sx)u_\bk(\sx')\,, 
\end{align}
the positive-frequency part $KEf$ of $Ef$ reads
\begin{align}
    KEf &= \int\dd^n\bk \,\braket{u_\bk,Ef}_{\textsf{KG}}u_\bk \equiv \int\dd^n\bk\,{f_\bk} u_\bk  
\end{align}
where
\begin{align}
    {f}_\bk\coloneqq {-\ii}\int\dd V\,f(\sx)u_\bk^*(\sx)\,.
    \label{eq: f-fourier}
\end{align} 

Furthermore, since $KEf$ is an element of the one-particle Hilbert space $\mathcal{H}$, it is convenient to use the notation $\ket{KEf}\in \mathcal{H}$ and the (improper) basis as $\ket{u_\bk}$. In this notation, we write
\begin{align}
    \ket{KEf} = \int\dd^n\bk \, {{f}_\bk} \ket{u_\bk}\,,
    \label{eq: KEf-notation}
\end{align}
This suggests that it is useful to instead label the ladder operators using the ``momentum space smearing function'' ${f}_\bk$ rather than the real space function $KEf$. For this reason we will often write the argument with respect to the $\bk$-space counterpart:
\begin{align}
    \hat{\phi}(f) &= \hat{a}({f}_\bk^*) + \hat{a}^\dagger({f}_\bk) \equiv \int\dd^n\bk\,\rr{\hat{a}_\bk^{\phantom{\dagger}} {{f}^*_\bk} + \hat{a}^\dagger_\bk {{f}_\bk}} \notag\\ 
    &\equiv \hat{\phi}(f_\bk)\,.
    \label{eq: fock-representation}
\end{align}
Using $f_\bk$ essentially means that we are using the elements of the one-particle Hilbert space $\mathcal{H}$ to label the field observables rather than using the set of spacetime smearing functions $\CS$. For many practical calculations such as time evolution of observables this relabeling is also more natural as they are adapted to Fock space structure (see Section~\ref{sec: consequences}).

We remark that the textbook version of canonical quantization refers to the Fock representation that is unitarily equivalent to the \textit{GNS representation of the vacuum state} $\omega_0$ (which is quasifree), which we call the \textit{vacuum representation}. In this representation, the GNS cyclic vector is the Fock vacuum $\ket{0}$ satisfying $\hat{a}_\bk\ket{0}$ for all $\bk$. Importantly, not all quasifree representations are unitarily equivalent even if they are all Fock representations. The simplest example in quasifree setting is the \textit{thermal representation} obtained from the \textit{Kubo-Martin-Schwinger (KMS) states} \cite{kubo1957statistical,martinSchwinger1959theory}. It is well-known that for infinite systems obeying the CCR algebra, the GNS representations of two KMS states at different temperatures are not unitarily equivalent \cite{takesaki1970disjointness,muller1980disjointness}, and in particular finite-temperature thermal representation will not be unitarily equivalent to the zero-temperature vacuum representation. This can be traced to the fact that the vacuum Fock representation only supports \textit{finitely many particles}, while the thermal representations can support \textit{finite particle density} but infinite particle content\footnote{One encounters this when calculating the particle content in using the Bogoliubov transformation approach in the context of Unruh effect \cite{birrell1984quantum}. This is the statement that the Fock representations of the Rindler vacuum and the Minkowski vacuum are unitarily inequivalent, and with respect to the Rindler time  the Minkowski vacuum is a KMS state \cite{bisognano1975duality}.}. This observation is relevant for us since in the UDW context and interacting QFT more generally, some problems arise when we try to shoehorn the interacting ground state into the vacuum  representation of the free theory.

\subsection{Ultraviolet and infrared problems in bosonic field theory}
\label{sec: IR-divergences}

Since we are interested in UDW as a simplified model of light-matter interaction, it is imperative that what we do applies  (at least) for a massless scalar field with linear dispersion $\omega_\bk = |\bk|$. As is well-known to the mathematical physics community, this turns out to be a non-trivial problem. We will use this section to illustrate the IR problem that plagues massless fields. In fact, an analogous problem already arises in the context of \textit{van Hove Hamiltonians} that describe a scalar field theory driven by classical source.

The van Hove model is an exactly solvable model for a scalar field theory in Minkowski spacetime driven by a classical {time-independent} external source $J$: the equation of motion is given by
\begin{align}
    (\partial_a\partial^a-m^2)\phi(t,\bx) = J(\bx)\,.
    \label{eq: van-hove-KGE}
\end{align}
The Hamiltonian at $t=0$ is given by\footnote{{Note that while the classical wave equation \eqref{eq: van-hove-KGE} can be solved for time-dependent $J$, the resulting Hamiltonian will be time-dependent and this does not allow us to speak of the ground state of the model without further assumptions on the time-dependence.}}
\begin{align}
    \hat{H}_J = \int \dd^n\bk\rr{\omega_\bk^{\phantom{\dagger}} \hat{a}_\bk^\dagger\hat{a}_\bk^{\phantom{\dagger}}  + \hat{a}_\bk^{\phantom{\dagger}}  z^*_\bk + \hat{a}_\bk^\dagger z_\bk^{\phantom{*}} } + C\openone \,.
\end{align}
Here $C$ is an arbitrary (possibly infinite) constant and the complex function $z_\bk$ is related to the Fourier transform of $ J(0,\bx)$ via
\begin{align}
    z_\bk =  \frac{\Tilde{J}_\bk}{\sqrt{2(2\pi)^n\omega_\bk}}\,,
\end{align}
where $\tilde{J}_\bk$ is the spatial Fourier transform of $J(\bx)\equiv J(0,\bx)$:
\begin{align}
    \Tilde{J}_\bk \coloneqq \int\dd^n\bx\,J(\bx) e^{-\ii\bk\cdot\bx}\,.
\end{align}
It is worth noting that the van Hove Hamiltonian can be considered as a special case of the Pauli-Fierz Hamiltonian \cite{derezinski2003vanHove}.

There are two kinds of van Hove model. First set $C=0$. This Hamiltonian defines a family of self-adjoint operators if the condition holds \cite{derezinski2003vanHove}:
\begin{align}
    \int_{B_0}\dd^n\bk\, |z_\bk|^2 + \int_{\R^n\setminus B_0}\hspace{-0.75cm} \dd^n\bk \,\frac{|z_\bk|^2}{\omega_\bk^2} <\infty
    \label{eq: IR+UV-condition}
\end{align}
where $B_0\coloneqq \{\bk\in \R^n:\omega_\bk\leq \omega_0\}$ for some fixed positive value $\omega_0$. Note that the first term imposes some IR regularity while the second term is a statement about the ultraviolet (UV) regularity. One way to see this is to observe that the first term of Eq.~\eqref{eq: IR+UV-condition} is always finite for massive fields ($\omega_\bk=\sqrt{|\bk|^2+m^2}$) even if $J(\bx)=\delta^n(\bx)$, while the second term would diverge. Similarly, for massless fields ($\omega_\bk=|\bk|$) the second term will be finite for $z_\bk$ that decays sufficiently fast, such as $e^{-\alpha^2|\bk|^2}$ for some $\alpha>0$, but the first term can be divergent if $z_\bk\sim |\bk|^{-\alpha}$ for sufficiently large $\alpha$.

Suppose instead that we pick
\begin{align*}
    C=\int\dd^n\bk\frac{|z_\bk|^2}{\omega_\bk} \neq 0\,,
\end{align*}
which allows us to complete the square and define the second type of van Hove model
\begin{align}
    \hat{H}_{J}' &= \int\dd^n \bk\,\omega_\bk^{\phantom{\dagger}}\hat{b}_\bk^\dagger\hat{b}_\bk^{\phantom{\dagger}}\,,\qquad \hat{b}_\bk\coloneqq \hat{a}_\bk+\frac{z_\bk^*}{\omega_\bk}\openone \,.
\end{align}
This Hamiltonian is self-adjoint if a stronger IR regularity condition is satisfied:
\begin{align}
    \int_{B_0}\dd^n\bk\, \frac{|z_\bk|^2}{\omega_\bk} + \int_{\R^n\setminus B_0}\hspace{-0.75cm} \dd^n\bk \,\frac{|z_\bk|^2}{\omega_\bk^2} <\infty
\end{align}
In particular, if we impose that
\begin{align}
    \int\dd^n\bk\,\frac{|z_\bk|^2}{\omega_\bk} <\infty\,,
    \label{eq: IR-regular-full}
\end{align}
then the two Hamiltonians only differ by $C$, otherwise $C=\infty$ in which case the first van Hove model has a UV divergence. However, the second model can still be well-defined by  an ultraviolet renormalization --- using an infinite counterterm to subtract $C$ \cite{derezinski2003vanHove}.

We can organize these results by first defining the following integral
\begin{equation}
    \begin{aligned}
    R^\Lambda_j(z_\bk,n)&\coloneqq\int_\Lambda\dd^n\bk\,\frac{|z_\bk|^2}{\omega_\bk^j}\,,\qquad \Lambda\subseteq \R^n\,,\\
     R^{\Lambda^c}_j(z_\bk,n) &\coloneqq R^{\R^n\setminus\Lambda}_j(z_\bk,n)\,.
    \end{aligned}
\end{equation}
We will write $R_j(z_\bk,n)$ when $\Lambda=\R^n$. We are interested in situations where there is no UV problem in $R_j(z_\bk,n)$, i.e., $z_\bk$ decays sufficiently fast at large $|\bk|$. If $R_{j}<\infty$ then $R_{i}<\infty$ for all $0\leq i < j$ since removing powers of $1/\omega_\bk$ always improves the IR and the UV remains well-regulated if $z_\bk$ decays superpolynomially. For our purposes we will assume the following UV regularity:
\begin{assumption}[UV regularity]
    \label{assumption: sensible}
    The choice of $z_\bk$ is such that 
    \begin{align}
        R_0(z_\bk,n) <\infty \,.
    \end{align}
\end{assumption} 
\noindent All known calculations in the UDW literature satisfy this requirement: in particular, this holds for the hydrogen-like atoms \cite{pozas2015harvesting}. This will be the bare minimum we need for the subsequent discussions.

As we will discuss in Section~\ref{sec: UDW-model}, while the UDW model is expected to have good UV behaviour due to the spatial localizability of the detector\footnote{In situations where both pointlike limit and sharp switching are convenient (see Section~\ref{sec: UDW-model} for the description of the UDW model), as is often the case in the master equation approaches \cite{kaplanek2020hot,floreanini2004unruh,moustos2017nonmarkov,tjoa2023effective}, a UV cutoff needs to be imposed by hand.}, Assumption~\ref{assumption: sensible} alone is not sufficient and some care is needed to control the IR behaviour of the model. In the standard interaction picture calculations, this issue is not visible as this concerns the existence of ground states (and hence the validity) of the model. We end this section by mentioning that the spin-boson Hamiltonian and van Hove Hamiltonians are special case of the Pauli-Fierz Hamiltonian which has also been extensively studied in the literature to varying degrees of generality (see, e.g.,  \cite{derezinski1999asymptotic,gerard2000existence,griesemer2001ground} and references therein).

\section{UDW model and the spin-boson model} 
\label{sec: UDW-model}

In this section we will discuss the UDW model in curved spacetime and its connection to the so-called \textit{spin-boson} model in the literature. We will see that there is a sense in which the two models involve the same kind of physical system but attempting to address two quite different physical scenarios, distinguished by the \textit{temporal localizability} of the interactions.

\subsection{Standard UDW detector model in curved spacetimes}

The global hyperbolicity of the spacetime $\M$ guarantees the existence of a global foliation $\M\cong \R\times \Sigma$ where $\Sigma$ is a Cauchy surface in $\M$. We have a coordinate system $(t,\bx)$ that labels the foliation so that we have a one-parameter family of Cauchy surfaces $\Sigma_t$ that defines ``constant-$t$ slices''. {In particular, this means that the function $t$ which labels these surfaces also defines their normals: $n_a=-N(\sx)\partial_at$, where $N(\sx)$ is the lapse function effectively determined by the normalization of $n$. We also assume that the congruence of curves connecting any two Cauchy surfaces have tangent vector $t^a$ so that $t^a\partial_a t=1$. }

In what follows, we will restrict our attention to only \textit{static spacetimes}, i.e., those that admit global timelike Killing vector field $t^a$ such that in the ``adapted coordinates'' we have $t^a\equiv (\partial_t)^a$ and the metric takes the time-independent form
\begin{align}
    \dd s^2 = -N(\bx)^2\dd t^2 + h_{ij}(\bx)\dd x^i\dd x^j\,,
\end{align}
where $h_{ij}$ is the spatial metric on each $\Sigma_t$. {We note in this case $t^a=N(\bx) n^a$.} We will assume also that there is no coupling to Ricci scalar curvature, i.e., $\xi=0$ in the Klein-Gordon equation \eqref{eq: KGE}.

The existence of the global timelike Killing vector guarantees that the free Hamiltonian of the scalar field is \textit{time-independent}. In more detail, {for a minimally coupled scalar field ($\xi=0$ in Eq.~\eqref{eq: KGE})} the stress-energy density for the scalar field theory in the adapted coordinates is given by 
\begin{align}
    T_{00}(\sx) =
    (\partial_t\phi)^2 + \frac{1}{2}N(\bx)^{2}\rr{\partial_i\phi\partial^i\phi + m^2\phi^2}\,,
\end{align}
hence the free Hamiltonian of the scalar theory at time $t$ is given by 
\begin{align}
    {H}_0 \equiv -\int_{\Sigma_t}T^{a b}t_{a}d\Sigma_b \,.
\end{align}
Setting $t=0$, the standard quantization procedure allows us to write the quantized Hamiltonian as ``normal-ordered'' operator
\begin{align}
    \hat{H}_0 &= \int\dd^n\bk\,\omega^{\phantom{\dagger}}_\bk\hat{a}_\bk^\dagger\hat{a}_\bk^{\phantom{\dagger}}\,,
    \label{eq: free-Hamiltonian}
\end{align}
which is invariant under time translations if the spacetime is static. 

Note that in generic curved spacetimes we cannot perform normal ordering as there is no ``preferred vacuum state'' to subtract the UV-divergent contribution. However, since physically reasonable states are required to be Hadamard states, one way to handle this is to consider \textit{point-splitting regularization}\footnote{This in itself is not without drawbacks --- see \cite{moretti2003comments}.}. In essence, the regularization works by first considering the ``point-split'' stress-energy density 
\begin{align}
    & T_{00}(\sx,\sy) 
    \coloneqq \partial_t\phi(\sx)\partial_t\phi (\sy) +\frac{1}{2}N(\bx)N(\bm{y})\rr{\partial_i\phi(\sx)\partial^{i}\phi(\sy) + m^2\phi(\sx)\phi(\sy)}\,,
\end{align}
and the stress-energy tensor is obtained in the coincidence limit $\sy\to\sx$. After quantization, the expectation value with respect to some Hadamard state is given by $ \omega(\hat{T}^{00}(\sx,\sy))$ which is highly singular in the coincidence limit. However, since the difference between two Hadamard states is a \textit{smooth} bi-distribution, the \textit{renormalized stress-energy density} will be well-defined and is defined via
\begin{align}
    \omega(\hat{T}_{00}^{\text{reg}}(\sx)) \coloneqq \lim_{\sy\to \sx}\omega(\hat{T}_{00}(\sx,\sy)) - \omega_0(\hat{T}_{00}(\sx,\sy))
\end{align}
where $\omega_0$ is some reference Hadamard state. This is morally equivalent to normal ordering in the sense that 
\begin{align}
    \hat{T}^{\text{reg}}_{00}(\sx) \equiv \lim_{\sy\to \sx}\hat{T}_{00}(\sx,\sy) -\omega_0(\hat{T}_{00}(\sx,\sy))\openone \,.
\end{align}
The point-splitting is ambiguous in the sense that there is no preferred reference Hadamard state, but for static spacetimes there exist a preferred vacuum state $\omega_0$ that is invariant under the Killing time translation. Once this is fixed, the resulting quantum Hamiltonian $\hat{H}_0$ will be given precisely by Eq.~\eqref{eq: free-Hamiltonian}.

The Hamiltonian of the commonly used UDW model takes the form of a \textit{time-dependent} Hamiltonian
\begin{align}
    \hat H(t)\coloneqq \hat{H}_0 + \hat{h}_0 + \hat{H}_I(t)\,,
    \label{eq: full-UDW-Hamiltonian}
\end{align}
where $\hat{h}_0=\Omega\hat{\sigma}_z + \Delta\hat{\sigma}_x$ is the free Hamiltonian of the qubit detector with $\Omega\geq 0$ and $\Delta\in \R$. The interaction Hamiltonian is defined to be
\begin{align}
    \hat{H}_I(t) &= \lambda\int_{\Sigma_t}\dd^n\bx\sqrt{h}\,f(\sx)\hat{\sigma}_x(\tau(\sx))\otimes \hat{\phi}(\sx)\,.
\end{align}
Here $\lambda$ is the coupling constant, $f(\sx)$ is the spacetime smearing function prescribing the interaction region between the qubit detector and the field, and $\tau$ is the proper time of the qubit's {center of mass} (COM). 

In principle, one can then proceed to the dynamics of the system, but the dependence of $\tau$ on $\sx$ makes any practical calculations difficult without further simplifications such as the form of $f$. We need a coordinate system that is adapted to the neighbourhood of the COM trajectory of the qubit (where the interaction occurs).  This is given by the Fermi normal coordinates (FNC) \cite{Tales2020GRQO,Bruno2020time-ordering} $(\tau,\bar{\bx})$, which is the coordinate system where the COM trajectory is labelled by $\bar{\bx}=0$ and neighbouring points at fixed $\tau$ have proper distance $r=\sqrt{\Bar{\bx}\cdot\Bar{\bx}}$. Physically, this is the comoving frame of the observer moving along some fixed trajectory $\mathsf{z}(\tau)$ in curved spacetime.

One physical input we need to make progress is to assume that in the FNC, the spacetime smearing function \textit{factorizes} into spatial and temporal part:
\begin{align}
    f(\sx(\bar{\sx}))\equiv \chi(\tau)F(\Bar{\bx})\,,
\end{align}
which corresponds to the statement that in the rest frame of the observer carrying the qubit they must be able to distinguish the spatial profile of the qubit from the ``knob'' $\chi$ (called the \textit{switching function}) that turns on and off the interaction with the field. 

Once we make this assumption, the interaction Hamiltonian reads
\begin{align}
    \hat{H}_I(\tau) &= \lambda\int_{\mathcal{E}_\tau}\dd^n\bar{\bx}\sqrt{\bar{h}}\,\chi(\tau)F(\bar{\bx})\hat{\sigma}_x(\tau) \otimes \hat{\phi}(\sx(\bar{\sx}))\,.
\end{align}
This is slightly simpler because $\hat{\sigma}_x$ only depends on $\tau$, which is analogous to the pointlike model in the original model \cite{Unruh1979evaporation,DeWitt1979}. However, for generic interactions we still run into a problem because for generic trajectories the timelike vector $\partial_\tau$ is not parallel to the Killing vector $\partial_t$. This means that Eq.~\eqref{eq: full-UDW-Hamiltonian} will have mixed time parameter that is not uniformly convenient: if we use $\tau$, the free Hamiltonian is not time-translation invariant with respect to $\tau$, and if we use $t$ the interaction Hamiltonian does not factorize cleanly into spatial and temporal part. In principle, one is free to quantize the \textit{total} system in terms of the proper time $\tau$, but {this is generically intractable since the wave equation may not admit simple expressions in arbitrary FNC coordinates.}

\subsection{Temporal localization problem}
The standard formulation of the spin-boson (SB) model is given by the following Hamiltonian \cite{amann1991spinboson,spohn1989spinboson,fannes1988equilibrium,hasler2011ground,hasler2021existence}
\begin{align}
    \hat{H}_{\text{SB}} &= \int\dd^n\bk\,\omega_\bk^{\phantom{\dagger}}\hat{a}_\bk^\dagger\hat{a}_\bk^{\phantom{\dagger}} + \Omega\hat{\sigma}_z + \Delta\hat{\sigma}_x + \lambda \hat{\sigma}_x\otimes \int\dd^n\bk\,g_\bk^{\phantom{\dagger}}\rr{\hat{a}_\bk^{\phantom{\dagger}}+\hat{a}_\bk^\dagger} \,.
    \label{eq: spin-boson-hamiltonian}
\end{align}
 Here $\Omega$ is the bare energy gap, $\Delta$ can be viewed as the detuning parameter, $\lambda$ is the coupling strength, and $g_\bk$ is traditionally called the coupling function which defines the \textit{spectral density} of the interaction. The special case where there is only one mode in the bosonic sector is known as the \textit{quantum Rabi model}\footnote{ The quantum Rabi model has been solved \textit{exactly} \cite{xie2017quantum,braak2011rabi}.}, whose Hamiltonian reads
\begin{align}
    \hat{H}_{\textsc{R}} &= \omega\hat{a}^\dagger\hat{a}^{\phantom{\dagger}} + \Omega\hat{\sigma}_z+\Delta\hat{\sigma}_x + g \hat{\sigma}_x\otimes \rr{\hat{a}+\hat{a}^\dagger} \,.
    \label{eq: rabi-hamiltonian}
\end{align}
where the coupling function $g$ is now a single number. It is clear that the UDW Hamiltonian \eqref{eq: full-UDW-Hamiltonian} is structurally the same as spin-boson Hamiltonian or the Rabi Hamiltonian in that they involve linear coupling between a spin observable and a bosonic field operator. 

Here is the crucial observation: the spin-boson Hamiltonian \eqref{eq: spin-boson-hamiltonian} (or the Rabi Hamiltonian \eqref{eq: rabi-hamiltonian}) is viewed as a {time-independent Hamiltonian} and hence we are free to define it with respect to some fixed time slice, i.e., a `time-zero' operator. The resulting time evolution operator is then given by unitary time evolution
\begin{align}
    \hat{U}_{\text{SB}}(t) &= e^{-\ii t\hat{H}_{\text{SB}}}
\end{align}
with
\begin{align}
    U_{\text{SB}}(t+s) = U_{\text{SB}}(t)U_{\text{SB}}(s)\,.
\end{align}
Notice that this property fails for time-dependent Hamiltonians. That is, in the language of the UDW framework the spin-boson Hamiltonian \eqref{eq: spin-boson-hamiltonian} has an interaction that is \textit{switched on at all times}, and this does not fit the standard UDW framework where the interaction to be switched on and off (preferably, adiabatically) by suitable choice of the switching function $\chi(\tau)$. In other words, the essential distinction between the spin-boson Hamiltonian and the UDW Hamiltonian is the  \textit{temporal localization} of the detector-field interaction. Observe that in the standard UDW framework we typically want to make a case for observers carrying detectors being localized \textit{in spacetime}, i.e., the spacetime smearing function $f$ is effectively compactly supported
\footnote{Gaussians can be regarded as compactly supported for all practical purposes if the tails are neglected for a given precision.}. 

In the case where we can factorize the spacetime smearing into switching function and spatial profile $f=\chi(\tau) F(\bar\bx)$, we simply get $\supp f = \supp(\chi)\times\supp(F)$. Note that if the spatial profile $F(\Bar{\bx})$ is strongly localized (e.g., compactly supported or exponentially decaying in $\bar{\bx}$), it would provide a natural UV cutoff for the detector model through the coupling function $g_\bk$ with strong Fourier decay. If one wishes to consider pointlike model, this is equivalent to imposing an ``experimental'' UV cutoff. In contrast, the spin-boson Hamiltonian amounts to setting the switching function $\chi(\tau) = 1$ for all $\tau\in \R$, and therefore $\supp(f) = \R\times \supp(F)$, which is at best only compactly supported along the spacelike direction. 

{From the physical point of view, the discrepancy between the standard UDW paradigm and the spin-boson (hereafter SB) paradigm has to do with what problems we are modeling and how we view the total system. On the one hand, the standard UDW paradigm views the problem as one where the detector and the field are initially independent subsystems, and we ask what happens to the detector and/or the field if we \textit{externally} switch on and off the interaction for a given choice of interaction Hamiltonian. In particular, this is a ``semiclassical'' model where there is an external classical control that modulates the interaction. On the other hand, the SB model (and also the van Hove model) views the detector and the field as a \textit{closed system}, hence by construction the joint system evolves with a time-independent Hamiltonian. Thus as a joint system there is no freedom to switch off the interaction, in the same way that in quantum electrodynamics we cannot demand that an electron does not respond to an electromagnetic field. } 

The difference between the two paradigms suggests that we view the notion of local observers as being distinct from having localized interactions between the detector and the field. For example, suppose that in the SB paradigm the interaction is localized along a worldtube $\mathsf{z}(U) \coloneqq \R\times U$ where $U$ is some compact spacelike subset and the COM trajectory $\sz(\tau)\subset \mathsf{z}(U)$. An observer that only has access to distant region $R$ can only measure observables supported in $R$. Even though the interaction lasts forever in this paradigm, the effect of the interaction decays very quickly with the geodesic distance between any point in $R$, so for all practical purposes the interaction is spatiotemporally localized in time. In the UDW paradigm, the temporal localization of the interaction means that it is possible for the interaction region $\supp(f)$ to be completely causally disconnected from $R$, but this requires a physical scenario where one would like to insist on being able to isolate, say, an atom from any electromagnetic environment. In this sense, the SB paradigm is closer to the way one views interaction in the Standard Model of particle physics where one is not free to switch off fundamental forces of nature\footnote{{In the Standard Model formalism, the Lagrangian that by construction the interaction is never switched off \textit{anywhere} in spacetime, while in the SB model it is at least localized along the spacelike direction. There the \textit{cluster decomposition theorem} guarantees that physical processes between two sufficiently well-separated regions cannot influence one another \cite{weinberg1995quantum}, so the effect of interactions are effectively localized spatiotemporally.}}.

What about localization of the observers? If we adopt the SB paradigm, it is necessary that a local observer is viewed as as an agent who can only access a restricted set of field observables through the field. Recall that the interaction occurs along the timelike worldtube $\sz(U)\cong \R\times U$ where $U$ is the spatial localization of the detector around the COM trajectory $\sz(\tau)$. A local observer having access to spacetime region $R$ can only access field observables $\hat{\phi}(f)$ where $\supp(f)\subset R$, and \textit{a priori} $R$ does not need to intersect or coincide with $\sz(U)$. Indeed, if the observer moves along some \textit{finite-time} trajectory $\gamma(\tau)$ with $\tau\in [\tau_0,\tau_1]\subset \R$, then the accessible algebra $\A(\gamma(\tau))$ coincides with the algebra of the \textit{timelike envelope}\footnote{The timelike envelope of a timelike curve $\gamma$ between two points is the set points that can be reached by smooth deformations of $\gamma$, keeping the endpoints fixed, such that the deformed curves are timelike, see Figure~\ref{fig: envelope}.} $\A(\mathcal{E}(\gamma))$ \cite{witten2023algebras}. Note that even if $\gamma(\tau)\subset \sz(\tau)$, the timelike envelope $\mathcal{E}(\gamma)$ will not be contained in $\sz(U)$ unless the two endpoints are very close to one another. If $R$ is sufficiently far away from $\sz(U)$ then the expectation values of the field observables in $R$ can be well-approximated by the values when the interaction is not present: for example, one can check that the ground state expectation value of the field sector is well-approximated by the vacuum expectation value of the free theory if $R$ is spatially well-separated from $\sz(U)$.

\begin{figure}[tp]
    \centering
    \includegraphics[scale=0.5]{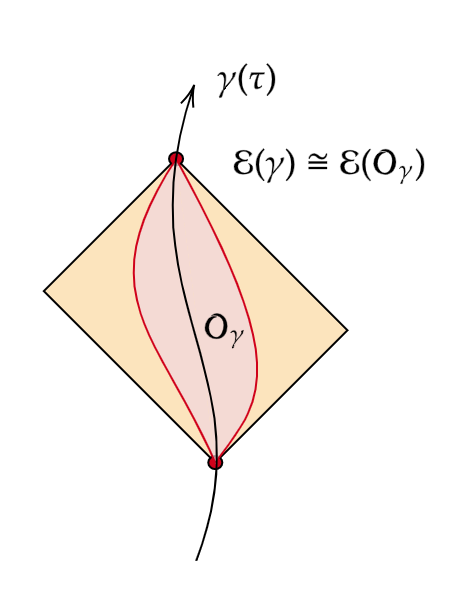}
    \caption{The timelike tube and envelope along a timelike trajectory $\gamma(\tau)$ with endpoints given by the red dots. {Here $\mathcal{O}_\gamma$ is the set of timelike trajectories with the same fixed endpoints that can be obtained by smooth deformation of $\gamma$. The timelike envelope of the curve $\mathcal{E}(\gamma)$ is the set of all points reachable by such a deformation.} Note that in general, even if $\gamma(\tau)$ coincides with $\mathsf{z}(\tau)$ in the sense that $\gamma(\tau)\subset \sz(\tau)$, the timelike envelope $\mathcal{E}(\gamma)$ will not be contained in the support of the UDW interaction $\R\times U$, thus the algebra of observables associated with the observer generically does not coincide with the support of the UDW interaction profile.
    }
    \label{fig: envelope}
\end{figure}

The temporal localization problem suggests that we should not associate the interaction region $\supp(\chi(\tau)F(\bx))$ of the UDW model with the localization of the observers since the latter is specified by the algebra of observables $\mathcal{A}(\gamma(\tau))$. This allows us to consider the UDW model in the SB paradigm where the interaction is always active ($\chi(\tau)=1$ for all $\tau\in\R$) without losing the notion of local observers who only have access to a restricted set of observables, namely those supported in the timelike envelope $\mathcal{E}(\gamma)$. This brings us closer to the standard quantum-optical formulation of a two-level system coupled to a bosonic field, and crucially it makes sense to view the joint system as being essentially a closed system. In what follows we will thus view the UDW model in the SB paradigm where the joint detector-field Hamiltonian is given by a \textit{time-independent} Hamiltonian analogous to the spin-boson model \eqref{eq: spin-boson-hamiltonian}, i.e.,  
\begin{align}
     \hat{H}_{\text{UDW}} &= \int\dd^n\bk\,\omega_\bk^{\phantom{\dagger}}\hat{a}_\bk^\dagger\hat{a}_\bk^{\phantom{\dagger}} + \Omega\hat{\sigma}_z+\Delta\hat{\sigma}_x + \lambda \hat{\sigma}_x\otimes \int\dd^n\bk\,{F}_\bk^{{*}}\hat{a}_\bk^{\phantom{\dagger}}+{F}_\bk^{\phantom{*}}\hat{a}_\bk^\dagger\,,
    \label{eq: full-UDW-hamiltonian-new}
\end{align}
where the difference with the SB model is that instead of considering an \textit{ad hoc} coupling function  $g_\bk\in L^2(\R^n)$, the UDW coupling function is obtained from a ``generalized Fourier transform'' of the spatial smearing function:
\begin{align}
    {F}_\bk\coloneqq \int_{\Sigma_{0}}\dd^n\bx\sqrt{h_0}\,F(\bx) u_\bk^*(0,\bx)
\end{align}
where $h_0$ is the induced metric at $t=0$ Cauchy slice $\Sigma_0$. In what follows we will call the time-independent UDW Hamiltonian as the UDW Hamiltonian, and the usual time-independent case as the `standard' UDW Hamiltonian. 

\subsection{Existence of a ground state and the UV/IR behaviour of the UDW model}

For the standard UDW Hamiltonian \eqref{eq: full-UDW-Hamiltonian}, the notion of ground state cannot generically be defined for all times since the Hamiltonian is time-dependent. Even in situations where the interaction is adiabatically switched on and off, the construction of the vacuum state only makes sense at the asymptotic past or future. In such a setup, the validity of the standard UDW model has to be viewed in the language of scattering theory. In this paper we will not pursue this route since we are not interested in scattering processes. 

Moreover, there are also subtleties involved in scattering theory for massless fields associated with the notion of asymptotic completeness, where this subject has a long and complicated history (see, e.g., \cite{derezinski1999asymptotic,derezinski2004scattering,DeRoeck2015SBscatter,dybalski2019scattering}). {In fact, adiabatic interactions in QFT do not remove the problem because the well-known adiabatic theorem \cite{barrysanders2004adiabatic} requires that the appropriate $S$-matrix is well-defined, but this is not guaranteed when IR divergences occur, even if we disregard the subtlety in correctly applying adiabatic-theorem type of arguments \cite{barrysanders2004adiabatic}. In any case, the standard UDW paradigm was not primarily meant to study scattering physics, so in principle the SB paradigm is closer to the conventional way the standard UDW model is used.}

In contrast, the (time-independent) UDW Hamiltonian \eqref{eq: full-UDW-hamiltonian-new} allows us to define a joint ground state that also respects the time-translation symmetry of the full Hamiltonian. If the UDW model is to be valid as a physical model, it must be the case that the Hamiltonian \eqref{eq: full-UDW-hamiltonian-new} describing the model has a stable ground state. This question is also relevant for the standard UDW model because if the UDW Hamiltonian \eqref{eq: full-UDW-hamiltonian-new} does not admit a ground state, we cannot expect the standard UDW Hamiltonian \eqref{eq: full-UDW-Hamiltonian} to admit an \textit{instantaneous} ground state either, since both expressions agree at $t=0$. {Consequently, while the validity of the (time-independent) UDW model does not rigorously guarantee the validity of the standard UDW model, the non-existence of the ground state in the UDW model immediately implies that the standard UDW model cannot be a valid physical model.}

To answer this question, first let us denote the UDW Hamiltonian \eqref{eq: full-UDW-hamiltonian-new} by $\hat{H}_{\Omega,\Delta,\lambda}$ to make the parameter dependence explicit, and set $\Omega=0$. 
\begin{proposition}
    \label{proposition: unitary-equivalence}
    The expectation value of the gapless qubit Hamiltonian $\hat{H}_{0,\Delta,\lambda}$ has a lower bound given by 
    \begin{align}
        \braket{\hat{H}_{0,\Delta,\lambda}} \geq -|\Delta| - R_1(\lambda F_\bk,n)\,.
    \end{align}
    In particular, the Hamiltonian is bounded from below if $R_1(\lambda F_\bk,n)$ is finite. If $\hat{H}_{0,\Delta,\lambda}$ admits a ground state $\ket{\psi_0}$ then the ground state energy is precisely
    \begin{align}
        E_0 =  - R_1(\lambda F_\bk,n) -|\Delta|\,.
    \end{align}
\end{proposition}

\begin{proof}
    Let $\ket{\pm}$ be the eigenbasis of $\hat{\sigma}_x$ with eigenvalue $\pm$ 1. Then we can write
    \begin{align}
        \hat{H}_{0,\Delta,\lambda} &= \ketbra{+}{+}\otimes \hat{H}_{+}+\ketbra{-}{-}\otimes \hat{H}_-\,,
    \end{align}
    with the field sector given by
    \begin{align}
        \hat{H}_\pm &= \hat{h}_0 \pm \rr{ \hat{a}(\lambda F^*_\bk) + \hat{a}^\dagger(\lambda F_\bk)} \pm \Delta\openone \notag\\
        &= \int\dd^n\bk \,\omega_\bk \hat{b}_{\bk,\pm}^\dagger\hat{b}_{\bk,\pm}^{\phantom{\dagger}} - \rr{\int\dd^n\bk\,\frac{\lambda^2|F_\bk|^2}{\omega_\bk}\mp\Delta} \openone \notag\\
        &\equiv \hat{h}_{0,\pm} - (\rr{R_1(\lambda F_\bk,n)\mp \Delta}\openone \,,
    \end{align}
    where the ladder operators are obtained by displacement:
    \begin{align}
        \hat{b}_{\bk,\pm} &= \hat{a}_{\bk} \pm \frac{\lambda F_\bk}{\omega_\bk}\openone \equiv \hat{D}(\pm \alpha_\bk)^\dagger\hat{a}_\bk\hat{D}(\pm \alpha_\bk)\,,
    \end{align}
    with $\alpha_\bk\coloneqq \lambda F_\bk/\omega_\bk$. Note that $\hat{h}_{0,\pm}$ is nothing but the van Hove Hamiltonian of the second type, which is unitarily equivalent to the standard free quadratic Hamiltonian provided the \textit{dressing unitary operator} $\hat{U}_\pm\coloneqq \hat{D}(\pm\alpha_\bk)$ exists:
    \begin{align}
        \hat{h}_{0,\pm} &= \hat{U}_\pm^\dagger\hat{h}_0\hat{U}_\pm\,,\quad  \hat{U}_\pm = e^{\pm\rr{\hat{a}^\dagger(\alpha_\bk)-\hat{a}(\alpha_\bk^*)}}\,.
    \end{align}
    If $\hat{U}_\pm$ exists, then the ground state of $\hat{h}_{0,\pm}$ is given by the coherent state $\ket{\mp \alpha}\equiv \hat{U}^\dagger_\pm\ket{0}$, where $\ket{0}$ is the Fock vacuum of the free Hamiltonian $\hat{h}_0$. The Hamiltonian $\hat{H}_{0,\Delta,\lambda}$ has lower bound attained by $\ket{\psi_0^\pm}\coloneqq \ket{\pm}\otimes\ket{\mp\alpha}$ depending on the sign of $\Delta$: if $\Delta>0$, then the ground state is $\ket{\psi^{-}_0}$, and hence
    \begin{align}
        \braket{\psi_0^-|\hat{H}_{0,\Delta,\lambda}|\psi_0^-} &= \braket{\alpha|\hat{h}_{0,-}|\alpha} - (R_1(\lambda F_\bk,n)+\Delta) \notag\\
        &=  -R_1(\lambda F_\bk,n) - \Delta
    \end{align}
    and similarly for $\Delta<0$ which corresponds to the ground state $\ket{\psi_0^+}$. Regardless of the sign of $\Delta$, the ground state has lowest energy
    \begin{align}
        E_0 = -R_1(\lambda F_\bk,n)-|\Delta|\,.
    \end{align}
    $E_0$ is finite if and only if $R_1(\lambda F_\bk,n)$ is finite. 
    
\end{proof}

Let us make some remarks about the Hamiltonian in Proposition~\eqref{proposition: unitary-equivalence}. First, if the detuning parameter or bias is removed $(\Delta=0)$ then the joint system has two-fold degenerate ground states. This follows since $\hat{H}_\pm$ now has the same ground state energy $E_0=-R_1(\lambda F_\bk,n)$ for both $\ket{\psi^\pm_0}=\ket{\pm}\otimes\ket{\mp\alpha}$. Second, if $E_0$ is finite then the inclusion of finite energy gap of the detector will only decrease the lower bound by at most $-\Omega$, so the Hamiltonian is still bounded below. However, as we have seen this depends essentially on whether $R_1(\lambda F_\bk,n)$ is finite. Proposition~\ref{proposition: unitary-equivalence} shows that if $R_1(\lambda F_\bk,n)$ is finite, then the UDW Hamiltonian with $\Omega=0$ is bounded from below with lowest eigenvalue $E_0$. 

{Finally, the theorem by itself does \textit{not} guarantee that the ground state lives in the tensor product of the \textit{non-interacting Hilbert space} $\C^2\otimes \fock$, where $\fock$ is the vacuum representation of the free field theory. It only says that \textit{if} it does live in $\C^2\otimes \fock$ then it takes the form $\ket{\pm}\otimes \ket{\mp \alpha}$ depending on the sign of the detuning parameter $\Delta$.} 

Proposition~\ref{proposition: unitary-equivalence} demands that we impose a constraint on the choice of spatial profile $F(\bx)$:
\begin{assumption}[Finiteness of ground state energy]
    \label{assumption: finite-energy}
    We assume that the spatial smearing function $F$ satisfies
    \begin{align}
        R_1(\lambda F_\bk,n) \equiv \int\dd^n\bk\frac{\lambda^2|F_\bk|^2}{\omega_\bk} <\infty\,.
    \end{align}
\end{assumption}
\noindent Unlike the SB model where this assumption would be viewed as a mathematical constraint on both the UV and the IR behaviour of the coupling function, in the context of the UDW model  Assumption~\ref{assumption: finite-energy} is mainly IR in nature. If we demand that the spatial profile $F$ is (effectively) compactly supported, then we expect $F_\bk$ to decay superpolynomially at large $|\bk|$. For hydrogenlike atoms in $(3+1)$-dimensional spacetime, $R_1$ will indeed be finite since at large $|\bk|$ the integrand of $R_1$ decays at least as fast as $|\bk|^{-2}$. Thus the finiteness of $R_1$ is mainly about what happens near $|\bk|=0$. Indeed, for massive scalar field the integral is always finite in the IR unless the spatial profile is chosen to exhibit additional IR singularity (e.g., by choosing $F$ such that $F_\bk$ contains some powers of $|\bk|^{-1}$). 

To make things more concrete in the massless case where $\omega_\bk=|\bk|$, let us consider the flat spacetime $\mathcal{M} = \R^{1,n}$. The spatial part has $F_\bk\sim |\bk|^{-1/2}\tilde{F}(\bk)$ where $\tilde{F}(\bk)$ is the Fourier transform of the spatial profile. For $R_1(\lambda F_\bk,n)$ to be IR-finite it is necessary that for some $\omega_0>0$ we have
\begin{align}
    \int_0^{\omega_0} \dd |\bk|\,\frac{|\tilde{F}(\bk)|^2}{|\bk|^{{3-n}}}<\infty\,.
    \label{eq: mild-IR-smearing}
\end{align}
In the pointlike limit (which does not modify the IR behaviour) this integral will be IR-finite for {all $n\geq 3$ but is IR-divergent for $n=1,2$}. In other words, a good UV behaviour alone is not sufficient to make the UDW model well-defined since the Hamiltonian can be unbounded from below (e.g., if we choose the spatial smearing to be Gaussian): one has to also ensure that the interaction profile regulates the IR. {In particular, for $n=3$ the UDW model with the commonly used choice of non-negative spatial profile (Gaussian, Lorentzian), $R_1$ is finite and hence the full Hamiltonian is bounded from below; in contrast, for compactly supported smearing functions (which have an additional $|\bk|^{-\alpha}$ in the Fourier space for some $0<\alpha<1$) is not guaranteed to be IR-finite.}

There are two natural ways to ensure that $R_1$ is IR-finite for the UDW model:
\begin{enumerate}[label=(\alph*),leftmargin=*]
    \item We can demand that the spatial smearing directly regulates the IR: this can be done by simply demanding that $F(\bx)$ has Fourier transform of the form
    \begin{align}
        \Tilde{F}(\bk)\sim |\bk|^a g(\bk)\,, \quad a\geq \frac{3-n}{2}\,,
    \end{align}
    where $\lim_{|\bk|\to 0}g(\bk)$ is finite. {This would imply that the spatial profile $F$ \textit{must have negative component} whenever $a>0$}. In itself this is not surprising since for the hydrogenlike atoms the corresponding spatial profile has the form $\bm{F}_{ij}(\bx)\sim \psi_i^*(\bx){\bx}\psi_j(\bx)$ associated with two energy levels $i,j$ that can have negative components (see, e.g.,  \cite{pozas2016entanglement}).

    \item We can adopt an ``experimental'' IR cutoff and regulate the IR using hard cut-off or through small mass parameter $m$ (effectively considering massive fields). IR cutoffs are natural given that in any realistic experiment the setup is always confined somewhere and any measurement devices involved in the experiment has an IR cutoff scale\footnote{Soft bosons are physically relevant in different settings. For instance, the seminal paper by Low \cite{low1958soft} suggests that soft photon radiation occurs in high-energy collisions and it turns out that there are still large discrepancy between the theoretical predictions and experimental results (see, e.g., \cite{wong2014overview}). Another example is in the study of Bose-Einstein condensation \cite{derezinski2004scattering}.}. Note that if we impose a hard IR cutoff directly, then the UDW model is automatically IR-finite.
\end{enumerate}
In what follows we focus on (a) since the standard UDW model does not automatically assume the existence of an IR cutoff. That is, we would like to know under what circumstances the spatial profile of the interaction leads to a well-defined UDW model. 

The important insight that can be extracted from the SB literature (see \cite{fannes1988equilibrium,amann1991spinboson,spohn1989spinboson}) is that even if $R_1(\lambda F_\bk,n)$ is finite, the ground state of the UDW model is not necessarily a state in the standard Fock space of the free theory. This occurs, in particular, if $R_2(\lambda F_\bk,n)=\infty$.
As we saw earlier from the van Hove Hamiltonian in Section~\ref{sec: IR-divergences}, physically the divergence of $R_2(\lambda F_\bk,n)$ means that the ground state must contain infinitely many IR bosons ({soft photons} in the context of quantum electrodynamics). A state that has finite energy but infinitely many bosons cannot lie in the vacuum sector of the Fock representation since by construction the standard all states in the standard Fock space must have finite particle number\footnote{{There are other Fock representations that are not unitarily equivalent to the vacuum Fock representation (see \ref{appendix: araki-woods})}}. 

Consequently, having infinitely many soft bosons implies that {the interacting Hilbert space is not unitarily equivalent to $\C^2\otimes\fock$}: the two Hilbert space representations are disjoint.  One way to see this is to consider what happens if we ask for the probability of finding a finite  number of bosons with frequency up to some $\omega_0$, i.e., in a ball of radius $\omega_0$ centred at the origin of the momentum space $\R^n$. The probability has a Poisson distribution \cite{spohn1989spinboson}
\begin{align}
    \Pr(n(\Lambda)=N)=\frac{1}{N!}(\lambda^2R^\Lambda_2(\lambda F_\bk,n))^N e^{-\lambda^2 R^\Lambda_2(\lambda F_\bk,n)}
\end{align}
where $\Lambda = \{\bk\in \R^n:|\bk|<\omega_0\}$. Assuming good UV behaviour, we can extend $\Lambda\to \R^n$ and hence the probability distribution is finite if $R_2(\lambda F_\bk,n) <\infty$. 

Following the same argument that goes into Assumption~\ref{assumption: finite-energy} shows that $R_2(\lambda F_\bk,n)$ is IR-finite if
\begin{align}
    \int_0^{\omega_0}\dd|\bk|\,\frac{|\tilde{F}(\bk)|^2}{|\bk|^{{4-n}}}< \infty
    \label{eq: harsh-IR-smearing}
\end{align}
In other words, if we would like the Fock space of the free theory to contain states from the interacting Hilbert space, we need stronger IR regularity such as
\begin{align}
    \tilde{F}(\bk)\sim |\bk|^{a}g(\bk)\,, \quad a\geq \frac{4-n}{2}
    \label{eq: IR-finite-smearing}
\end{align}
for some function $g$ that  is finite in the IR. For example, again this automatically implies that for massless scalar fields, the spatial smearing of the UDW model \textit{cannot be a simple positive function on each Cauchy slice}: any $F$ whose Fourier transform is of the form \eqref{eq: IR-finite-smearing} must have negative component using the Fourier transform of $|\bk|^{a}$ and the convolution theorem. These results are stable in the sense that adding some finite energy gap will not remove the IR divergence.

From a more practical point of view, the above analysis suggests that for the specific case of massless field coupled to a qubit detector, one either has to use an ``experimental'' IR cutoff or sufficiently IR-regular spatial profile for the detector for the model to admit a ground state that lives in the standard Fock space $\C^2\otimes\fock$. The analysis of the van Hove model adapted to the SB or UDW model essentially shows that the Hilbert space representation for the interacting system is unitarily equivalent to the Hilbert space representation of the free theory (on which the interaction picture is based on) if and only if $R_2(\lambda F_\bk,n)$ is finite.  Physically it is the statement that the interacting Hilbert space is only unitarily equivalent to the free one if the ground state has finite energy \textit{and} particle number. Note that this sort of phenomenon is not surprising at all, since it is known that there exists many unitarily inequivalent representations of the CCR algebra in quantum statistical mechanics \cite{bratteli2002operatorv2} and this forms the basis for Unruh and Hawking effects.

\section{Revisiting the UDW model}
\label{sec: consequences}

In this section we apply the analysis we described earlier to revisit the phenomenology of the UDW model. The operator-algebraic framework is particularly suitable since it allows us to deal with the existence of unitarily inequivalent representations of the CCR algebra, and in particular what to do if we are in a situation where the free Fock space is not the correct interacting Hilbert space. We follow closely the exposition of \cite{spohn1989spinboson,fannes1988equilibrium,amann1991spinboson} with suitable modifications.

\subsection{Operator-algebraic framework for the UDW model}

First let us build the joint algebra for the UDW model. First, recall that by writing $\hat{\phi}(f)\equiv \hat{\phi}(f_\bk)$, i.e., by labeling operators through its $\bk$-space test functions, we get an equivalent $\bk$-space Weyl algebra $\tilde{\W}(\R^n)$ with Weyl relations
\begin{align}
    W(f_\bk)W(g_\bk) = W(f_\bk+g_\bk)e^{-\ii \Im(f_\bk,g_\bk)_\mathcal{H} }\,.
    \label{eq: k-space Weyl}
\end{align}
The inner product on the positive-frequency Hilbert space $\mathcal{H}$ is simply the re-writing of the original inner product $\braket{\cdot,\cdot}_\textsf{KG}$: using the notation \eqref{eq: KEf-notation}, we see that
\begin{align}
    \frac{1}{2}E(f,g) &= \Im\,\mathsf{W}(f,h) \equiv \Im\braket{KEf,KEg}_{\textsf{KG}}\notag\\
    &\equiv \Im\int\dd^n\bk\int\dd^n\bk' f_\bk^* g_{\bk'}^{\phantom{*}} \braket{u_{\bk}|u_{\bk'}} \notag\\
    &= \Im \int\dd^n\bk\, f_\bk^* g_{\bk}^{\phantom{*}} \equiv \Im(f_\bk,g_\bk)_{\mathcal{H}}\,.
    \label{eq: k-space-causal-propagator}
\end{align}
The last equality is nothing but the standard Hilbert-space inner product $L^2(\R^n)$, hence we identify
\begin{align}
    \mathsf{W}(f,g)\equiv \braket{KEf,KEg}_{\textsf{KG}} \equiv (f_\bk,g_\bk)_{\mathcal{H}}\,,
\end{align}
and regard $\mathcal{H}$ as a dense subspace of $L^2(\R^n)$. This relabeling of the Weyl algebra in terms of the $\bk$-space functions is useful in order to straightforwardly port the known results about spin-boson model.

Before we proceed, there is one important subtlety regarding the choice of spacetime smearing function $f$ and its $\bk$-space representation $f_\bk$. We advocated that the localization of the interaction should be distinguished from the localization of the observers and the observables that they can access. Consequently, even though the choice of coupling function $F_\bk$ in the UDW Hamiltonian \eqref{eq: full-UDW-hamiltonian-new} depends only on the spatial profile $F(\bx)$ and some restrictions are required to make the model well-defined (namely, Assumptions~\ref{assumption: sensible} and \ref{assumption: finite-energy}), this is independent of the choice of spacetime smearing functions for the field observables the observer is interested in measuring. That said, since the UDW Hamiltonian affects how the detector and the field observers evolve in time, we will see that this requires us to view $W(F_\bk)$ as an element of the $\bk$-space Weyl algebra $\tilde{W}(\R^n)$ even though $F(\bx)$ is not strictly speaking an element of $\CS$. The key that makes this work is that since the scalar field used to build the UDW model is a non-interacting field theory, spatially smeared field operators $\hat{\phi}(F)\equiv \hat{\phi}(F_\bk)$ with appropriate Fourier decay (such as compactly supported functions or strongly localized function such as Gaussian) are well-defined operators \cite{witten2023algebras}\footnote{This is one reason why some of the ``non-perturbative'' calculations in the UDW literature \cite{tjoa2022capacity,tjoa2022fermi,Simidzija2018no-go,Simidzija2020capacity,gallock-yoshimura2023tripartite} involving delta-coupled models have sensible results, as the Wightman functions remains finite even if the field operators are only smeared along the spatial directions.}.  

Given the CCR algebra $\tilde{\mathcal{W}}(\R^n)$ and the algebra of observables of a qubit which is a two-dimensional matrix algebra $M_2(\C)$, the joint algebra of observables is the $C^*$-algebra \cite{fannes1988equilibrium}
\begin{align}
    \mathfrak{B} \coloneqq M_2(\C)\otimes \tilde{\mathcal{W}}(\R^n)\,, 
\end{align}
which is unique as a $C^*$ tensor product. An arbitrary element $A$ of the joint algebra has the form
\begin{align}
    A= \begin{bmatrix}
        A_{11}& A_{12} \\ A_{21} & A_{22}
    \end{bmatrix}\,,\qquad A_{ij}\in \tilde{\mathcal{W}}(\R^n)\,.
\end{align}
An algebraic state on the joint algebra is the positive linear functional $\omega: \mathfrak{B} \to \C$ that has the form
\begin{align}
    \omega = \begin{bmatrix}
        \omega_{11} & \omega_{12} \\ \omega_{21} & \omega_{22}
    \end{bmatrix}\,,
\end{align}
where each $\omega_{ij}$ is a state for $\Wk$ and
\begin{align}
    \omega(A) = \sum_{i,j=1}^{2} \omega_{ij}(A_{ij})
\end{align}
and satisfies positivity and normalization conditions for all $X,Y\in \Wk$:
\begin{subequations}
    \begin{align}
        \sum_{j} \omega_{jj}(\openone) = 1\,,\qquad \omega_{jj}(X^\dagger X)\geq 0\,,\\
        |\omega_{12}(X^\dagger Y)| \leq \omega_{11}(X^\dagger X)\omega_{22}(Y^\dagger Y)\,,\\
        \omega_{12}(X) = (\omega_{21}(X^\dagger))^*\,.
    \end{align}
\end{subequations}
Note that in this language, the existence of degenerate ground state of $\hat{H}_{0,0,\lambda}$ can be seen from a purely symmetry argument: the total Hamiltonian has $\mathbb{Z}_2$-symmetry generated by the automorphism $\tau \in \text{Aut}(\mathfrak{B})$ such that $\tau^2=\openone_{\mathfrak{B}}$ and 
\begin{align}
    \tau(\hat{\sigma}_x) = -\hat{\sigma}_x\,,\quad \tau(W(F_\bk)) = W(-F_\bk)\,,
\end{align}
the latter formally equivalent to $\tau(\hat{a}_\bk) = -\hat{a}_\bk$ \cite{spohn1989spinboson,fannes1988equilibrium}.

{Let us make several comments on some conditions we need to impose on our states. As we saw in Section~\ref{sec: setup}, at the very least we would like our state $\omega$ to be a regular state, since only then it makes sense to speak of the generators of Weyl algebra and write formally $W(Ef)\equiv e^{\ii\hat{\phi}(f)}$, and furthermore in such a representation the map $t\mapsto W(tEf)$ will be continuous and $\{W(tEf):t\in \R\}$ define a strongly continuous one-parameter family of unitaries associated with the regular state \cite{bratteli2002operatorv2}. Indeed, the vacuum representation defines a regular state. However, for several technical reasons we need additional requirements, though some are naturally afforded by quasifree states. We also need to impose some requirements coming from the background geometry itself. We briefly mention these below. 

Starting from the relativistic consideration, we need at least two additional conditions. First, we saw earlier that in relativistic QFT, a physical requirement is that the state is a Hadamard state \cite{KayWald1991theorems,Radzikowski1996microlocal,wald1994quantum}, which enforces the correct short-distance singularity behaviour of the correlation functions and the finiteness of the renormalized stress-energy density. Second, we need to assume that there exists no zero mode \cite{KayWald1991theorems}. In the context of massless scalar field in a two-dimensional Einstein cylinder, it is known that there is a zero mode \cite{EMM2014zeromode,Tjoa2019zeromode} that admits no Fock representation, which leads to the ambiguity of the ground state or KMS state \cite{crawford2022lorentzian}. 

From the quantum-mechanical side, regular states are not sufficient to define creation and annihilation operators since their domains do not coincide with the field operators $\hat{\phi}(f)$. So additional conditions are needed if we wish to calculate expectation values of the form
\begin{align}
    \omega_{ij}(\hat{a}^{\#}(g_{1})...\hat{a}^{\#}(g_{n}))\,,
    \label{eq: ladder-expectation}
\end{align}
for all $g_{j}\in \mathcal{H}$, where $\hat a^\#$ refers to the creation or annihilation operators\footnote{Formally  each $\hat{a}(f_\bk)$ can be viewed as the complex-smeared field operator $\frac{1}{\sqrt{2}}\hat{\phi}(f_\bk+\ii \omega_\bk f_\bk)$.}. A simple solution to this issue is to consider \textit{analytic states}, namely those where \cite{bratteli2002operatorv2}
\begin{align}
    t\in\R\mapsto \omega_{ij}(W(tf_\bk))
\end{align}
{is analytic in $t$ and hence $\sigma$-weakly (and strongly) continuous} \cite{amann1991spinboson,spohn1989spinboson,fannes1988equilibrium}.

In particular, this implies that the expectation value $\omega(W(tf_\bk))$ can be computed using series expansions in $t^n\braket{\Omega_{\omega}|\pi_\omega(\hat{\phi}(f))^n|\Omega_\omega}$, hence the knowledge of $n$-point functions specify completely the properties of the analytic state. For analytic states the expectation value \eqref{eq: ladder-expectation} is well-defined. From analyticity we also obtain $\sigma$-weak convergence in expectation values \cite{amann1991spinboson}, i.e., 
\begin{align}
    \lim_{n\to\infty} \tr(\hat \rho \pi(W(f_{\bk,n}))) = \tr(\hat \rho W(f_{\bk}))
\end{align}
for any sequence $(f_{\bk,n})_{n=1}^\infty$ converging to $f_\bk$ and for all density operators $\hat \rho$ {associated with states in the folium of $\pi$}. 

We remark that quasifree states form a particularly nice class of analytic states and the GNS representation of quasifree states gives rise to a Hilbert space with the familiar Fock space structure \cite{fewster2019algebraic}. Apart from the vacuum representation which we are most familiar with, two other quasifree representations that are well-known and suitable for our discussion in this section are the \textit{thermal representation} and the \textit{Araki-Woods representation} (see \ref{appendix: araki-woods}).}

\subsection{Dynamics}
\label{sec: dynamics}

In the algebraic framework, time evolution can be described by specifying what is known as a \textit{dynamical system}. We follow the exposition in \cite{Bratteli1972afalgebra,pillet2006open}.

We say that we have a \textbf{$C^*$ dynamical system} if we are given a triple $(\mathfrak{B},G,\alpha)$ where $\mathfrak{B}$ is a $C^*$ algebra, $G$ is a locally compact group, and $\alpha\in \text{Aut}(\mathfrak{B})$ is an automorphism of $\mathfrak{B}$ that is also a \textit{strongly continuous} representation of $G$:
\begin{align}
    \alpha_e = \Id_{\mathfrak{B}}\,,\quad \alpha_g\alpha_h=\alpha_{gh}
\end{align}
for all $g,h\in G$, $e$ the identity element of $G$, and $\Id_{\mathfrak{B}}$ is an identity map on $\mathfrak{B}$. Note that the map $g\mapsto \alpha_g(A)$ is continuous for all $A\in\mathfrak{B}$. 

We say that we have a \textbf{$W^*$ dynamical system} if $\mathfrak{B}$ is a von Neumann algebra\footnote{Historically $W^*$ algebra refers to a von Neumann algebra.} and $\alpha$ is \textit{weakly continuous} representation of $G$. For our purposes, we will be interested in the case when $G=\R$, so that $\alpha_t$ implements time evolution of observables in $\mathfrak{B}$ and the automorphism is associated with time translations in the sense that $\alpha_{t_1}\alpha_{t_2} = \alpha_{t_1+t_2}$ since $G$ is Abelian. 

A \textbf{covariant representation} of a dynamical system is a triple $(\mathcal{H},\pi, U)$ where $\mathcal{H}$ is a Hilbert space, $\pi$ is a non-degenerate representation of the $C^*$ algebra $\mathfrak{B}$ acting on $\mathcal{H}$ (for von Neumann algebra we require it to also be a \textit{normal representation}\footnote{A normal representation maps is one that is associated to the GNS representation of normal states. Normal states are precisely the states that have convergence in expectation values.}), and $U$ is a strongly continuous unitary representation of $g\in G$ such that 
\begin{align}\label{eq: 1 param evo}
    \pi(\alpha_g(A)) = U_g^\dagger \pi(A)U_g\quad \forall A\in \mathfrak{B}\,.
\end{align}
In the context $G=\R$, a covariant representation implies that all automorphisms are implementable as a unitary time evolution in the representation.

Roughly speaking, having a dynamical system means we are given an algebra of observables associated with the physical system at hand and how to evolve these observables in time (via the map $\alpha_t$). In a $C^*$ dynamical system, we can perform time evolution of observables without reference to any notion of states or representations. Furthermore, if the quantum system is finite-dimensional, the $C^*$ algebra of observables is simply $M_n(\C)$ and its state space is naturally identified with finite-dimensional Hilbert space. In this case, it is possible to realize the time evolution as a unitary conjugation $\alpha_t(A) = U_t^\dagger A U_t$ and $U_t = e^{-\ii H t}$ for some bounded operator $H$. Crucially, it shows that the time evolution is given by the {Schr\"odinger equation {and is equivalent to working in the Heisenberg picture}.}

It turns out that there are situations where it is not possible to define a $C^*$ dynamical system with respect to some evolution map $\alpha_t$ and only $W^*$ dynamical system is obtained. This occurs, for example, when we consider a UDW detector with small energy gap $\Omega$ with Hamiltonian $H_{\Omega,\Delta,\lambda}$, where the \textit{perturbed dynamics} $\alpha_t^{\Omega}$ around the gapless detector model $\hat{H}_{0,\Delta,\lambda}$ does not make sense at the level of $C^*$ algebras as it may bring the observables outside the joint algebra (i.e., it is not an automorphism). We say that such dynamics is a \textit{pseudodynamics} \cite{amann1991spinboson}. The issue is also related to the fact that the automorphism is not continuous at the $C^*$ algebra level, hence one needs to pass to the $W^*$ setting where the evolution is continuous in the weak operator topology \cite{pillet2006open,jakvsic1996model}. 

Finally, we should mention that one can in principle formulate some dynamical map for time-dependent Hamiltonians.     Let $H(t)$ be a time-dependent Hamiltonian such that it satisfies the time-dependent Schr\"odinger equation
\begin{align}
    \ii \frac{\dd}{\dd t} \hat{U}(t,s) &= \hat{H}(t)\hat{U}(t,s)\,,
\end{align}
where $s,t\in \R$ and $\hat{U}(s,s)=\openone$. The two-parameter family of operators $\hat{U}(t,s)$ are called \textit{unitary propagators} that are formally given by the {time-ordered exponential}
\begin{align}
    \hat{U}(t,s)\coloneqq\mathcal{T}\exp\rr{-\ii\int_s^t\dd\tau\,\hat{H}(\tau)}\,.
    \label{eq: unitary-propagator}
\end{align}
Formally, the unitary propagators define a two-parameter family of automorphisms $\tau_{t,s}:\mathfrak{B}\to\mathfrak{B}$ such that
\begin{align}
    \tau_{t,s}(A) &= \hat{U}(t,s)^\dagger A\hat{U}(t,s)\,.
\end{align}
satisfying the standard composition property $\hat{U}(u,t)\hat{U}(t,s)=\hat{U}(u,s)$, but in general the unitary propagators fail the semi-group property $\hat{U}(t,s) = \hat{U}(t-s,0)$. Whenever $\hat{U}(t,s)$ exists, this gives us a generalization of the dynamical system for time-dependent Hamiltonians.

{The time evolution in the standard UDW model is given by the unitary propagator \eqref{eq: unitary-propagator}. However, working with $\hat{U}(t,s)$ rigorously is difficult especially in situations involving massless fields \cite{hubner1995radiative}. The unitary propagator has been rigorously shown for bounded time-dependent Hamiltonians or time-dependent bounded perturbations \cite[Thm. X.69]{reed1975vol2}, but these do not apply to the {time dependent} UDW model {since the interaction Hamiltonian is an unbounded operator}. There are exceptions, such as in the context of gapless qubit model where the Magnus expansion applies, but here one still needs some additional assumptions on $\hat{U}(t,s)$  \cite{batkai2012norm}.}

\subsection{Time evolution of observables}
\label{sec: time evolution}

In Section~\ref{sec: UDW-model} we argued that we should distinguish between the observables that a local observable can access from the support of the standard UDW interaction specified by some compactly supported spacetime smearing $f\in\CS$. That is, while the support of the UDW interaction is related to where the observer ``carrying'' the detector is located, the support of the observables in the algebra $\A(\mathcal{E}(\gamma))$ for the time envelope $\mathcal{E}(\gamma)$ will not generically be contained in $\supp f$. This distinction allows us to speak of local algebras of observables while extending the interaction region to $\R\times U$ where $U=\supp(F)$ is the support of the spatial profile of the detector-field interaction. To close this section, we show how the calculations of time evolution of observables can be done in the operator-algebraic framework, following the spirit of \cite{fannes1988equilibrium,amann1991spinboson}.

It suffices to calculate the time evolution for the Weyl generators and Pauli operators. From a mathematical standpoint the computation is in a sense trivial, but it is instructive to do them explicitly since they may provide guidance for further generalizations and the potential applications to studying quantum information-theoretic problems related to the UDW model. For the gapless model $\Omega=0$, the time evolution is unitarily implementable via Eq.~\eqref{eq: 1 param evo}. In particular, we are interested in computing the time evolution of the following observables:
\begin{align}
    \alpha_t(\hat{\sigma}^j)\,,\qquad \alpha_t(W(g))
\end{align}
where $j=x,y,z$ are the Pauli operators and $g\in \CS$ is some spacetime smearing function that the observer can specify. From these basic observables we can in principle generate all other observables of interest. 

Let us first show that $\alpha_t(\hat{\sigma}^x) = \hat{\sigma}^x$. This follows directly from the fact that
\begin{align}
    \alpha_t \equiv e^{\ii t \delta}\,,\quad \delta(\hat A)\coloneqq [\hat{H},\hat A]
\end{align}
for all $\hat{A}\in \mathfrak{B}$. Since the full Hamiltonian commutes with $\hat{\sigma}^x$, $\delta(\hat\sigma^x)=0$ and hence $e^{\ii\delta}(\hat{\sigma}^x) \equiv \openone_{\mathfrak{B}}(\hat{\sigma}^x) = \hat{\sigma}^x$. It remains to calculate $\hat{\sigma}^z$, since $\hat{\sigma}^y = \ii \hat\sigma^x\hat{\sigma}^z$. The idea is to first write, in the $\ket{\pm}$ basis
\begin{align}
    W(g)\equiv \openone_{\C^2}\otimes W(g) = \begin{bmatrix}
        W(g) & 0 \\ 0 & W(g) 
    \end{bmatrix}\,.
\end{align}
The full Hamiltonian decomposes into
\begin{align}
    \hat{H} \equiv \begin{bmatrix}
        h + \phi & 0 \\ 0 & h - \phi
    \end{bmatrix}
\end{align}
where we used the shorthand $h$ for the field's free Hamiltonian and $\phi\equiv \hat{\phi}(\lambda F_\bk)$ to reduce notational clutter. The time evolved operator is given by
\begin{align}
    \alpha_t(W(g)) \equiv \begin{bmatrix}
        e^{\ii t(h+\phi)}W(g)e^{-\ii t(h+\phi)} & 0 \\ 0 & e^{\ii t(h-\phi)}W(g)e^{-\ii t(h-\phi)}
    \end{bmatrix}\notag
\end{align}
It suffices to compute one of them since the other matrix element is obtained by the substitution $\phi\to-\phi$. 

There are several ways to do this, but one efficient way is to use Lie-algebraic technique \cite{gilmore2008lie}. Observe that since the scalar field is non-interacting, in the $\bk$-space each mode acts like an independent harmonic oscillator. Hence it is sufficient to work, formally, with one mode $\bk$. Consider a four-dimensional subalgebra $\mathfrak{S}$ of the general linear algebra for $3\times 3$ complex matrices $\mathfrak{gl}_3(\C)$, defined by upper triangular matrices of the form
\begin{align}
    \mathfrak{S}\coloneqq \left\{\sum_{j=1}^4\alpha_j T^j\equiv \begin{bmatrix}
        0 & \alpha_1 & \alpha_2 \\ 0 & \alpha_3 & \alpha_4 \\ 0 & 0& 0
    \end{bmatrix}\right\}
\end{align}
where $\alpha_j\in\C$ and $T^j$ are the Lie algebra generators. The four operators $\{\hat{n}_\bk\equiv \hat{a}_\bk^\dagger\hat{a}_\bk^{\phantom{\dagger}},\hat{a}_\bk,\hat{a}^\dagger_\bk,\openone \}$ satisfy the same Lie algebra relations as $\{T^j\}$ via the identification
\begin{align}
    T^1 &\to \hat{a}_\bk\,, \quad T^2 \to \openone \,,\quad T^3\to \hat{n}_\bk\,,\quad T^4 \to \hat{a}_\bk^\dagger\,.
\end{align}
This is a non-Hermitian Lie algebra representation of the four bosonic operators. The key idea is that since the four operators form a Lie algebra, we can ask for the solution of the following matrix equation:
\begin{align}
    e^{\ii t \sum_j\alpha_j T^j} &= \prod_{j=1}^4e^{\ii t \beta_j T^j}\,.
\end{align}
This solution always exists and can be found explicitly by direct computation.  By considering all the modes together, it is possible to prove the following:

\begin{lemma}\label{lemma: lie-algebra-factorization}
    We have
    \begin{align}
        e^{\ii t(h+\phi)} &= e^{\ii\theta(t)}e^{\ii t h}e^{\ii\hat{\phi}(\lambda F_\bk(t))}\,,
    \end{align}
    where 
    \begin{align}
        F_\bk(t) &= {-} \frac{\ii F_\bk}{\omega_\bk} (1-e^{{-}\ii\omega_\bk t})\,,\\
        \theta(t) &= \frac{\lambda^2|F_\bk|^2}{\omega_\bk^2 }\rr{\sin(\omega_\bk t)-\omega_\bk t}\,.
    \end{align}
\end{lemma}
\begin{proof}
    One strategy is to do this mode-by-mode, i.e., by instead doing a simpler formal computation for
    \begin{align}
        \hat{\mathsf{U}}_\bk \coloneqq e^{ \ii t (\omega_\bk\hat{n}_\bk + \lambda F_\bk^*\hat{a}_\bk^{\phantom{\dagger}} +  \lambda F_\bk^{\phantom{*}} \hat{a}^\dagger_\bk) }
    \end{align}
    noting that $[\hat{\mathsf{U}}_\bk, \hat{\mathsf{U}}_{\bk'}]=0$ for $\bk\neq \bk'$. We then ask for the solution of
    \begin{align}
        \hat{\mathsf{U}}_\bk = e^{\ii \theta_1 \openone}e^{\ii \theta_2\hat{n}_\bk}e^{\ii \rr{\theta_3\hat{a}_\bk^{\phantom{\dagger}} + \theta_4 \hat{a}^\dagger_\bk}}
    \end{align}
    where $\theta_j$ are dependent on $\lambda, F_\bk$ and $\bk$. The result follows immediately after showing that $\theta_4=\theta_3^*$ and then combine all the modes together.

\end{proof}

Using Lemma~\ref{lemma: lie-algebra-factorization}, it is now straightforward to calculate the matrix elements of $\alpha_t(W(g_\bk))$. As before, we relabel $W(g)\equiv W(g_\bk)$ and using Eq.~\eqref{eq: k-space Weyl} and Eq.~\eqref{eq: k-space-causal-propagator}, it follows from direct computation
\begin{align}
    e^{\ii t(h+\phi)}W(g_\bk)e^{-\ii t(h+\phi)}
    &= e^{\ii\theta(t)}e^{ \ii t h}e^{\ii\hat{\phi}(\lambda F_\bk(t))}W(g_\bk)e^{-\ii\hat{\phi}(\lambda F_\bk(t))} e^{-\ii t h}e^{-\ii\theta(t)} \notag\\
    &= e^{{-} \ii {\Im(2\lambda F_\bk(t),g_\bk)_{\mathcal{H}}}} e^{ \ii t h} W(g_\bk) e^{-\ii t h}\notag\\
    &= e^{{-}\ii {\Im(2\lambda F_\bk(t),g_\bk)_{\mathcal{H}}}}  W(g_\bk{e^{\ii\omega_\bk t}})\,.
\end{align} 
Writing $ \varphi(t) = {-}{\Im(2\lambda F_\bk(t),g_\bk)_{\mathcal{H}}}$, we get
\begin{align}
    \alpha_t(W(g)) &= \begin{bmatrix}
        e^{\ii\varphi(t)}W(g_\bk e^{\ii\omega_\bk t}) & 0 \\ 0 & e^{-\ii\varphi(t)}W(g_\bk e^{\ii\omega_\bk t})
    \end{bmatrix}\notag\\
    &\equiv W(g_\bk e^{\ii\omega_\bk t})\otimes e^{\ii\varphi(t) \hat{\sigma}^x}\,.
\end{align}
Next, let us calculate the time evolution of $\hat{\sigma}^z$. We have in the $\ket{\pm}$ basis
\begin{align}
    \alpha_t(\hat{\sigma}^z) 
    &= 
    \begin{bmatrix}
        0 & e^{\ii t(h+\phi)}e^{-\ii t (h-\phi)} \\ e^{\ii t(h-\phi)}e^{-\ii t (h+\phi)} & 0
    \end{bmatrix}\,.
\end{align}
As before only one of the matrix elements is sufficient by substituting $\phi\to -\phi$. Using Lemma~\ref{lemma: lie-algebra-factorization}, we get
\begin{align}
    e^{\ii t(h+\phi)}e^{-\ii t (h-\phi)} 
    &=W(2\lambda F_\bk(t))
\end{align}
and hence
\begin{align}
    \alpha_t(\hat{\sigma}^z) &= \begin{bmatrix}
        0 & W(2\lambda F_\bk(t)) \\ W(-2\lambda F_\bk(t)) & 0
    \end{bmatrix}\notag\\
    &\equiv e^{\ii \hat\phi(2\lambda F_\bk(t))}\otimes \hat{\sigma}^+ + e^{-\ii \hat\phi(2\lambda F_\bk(t))}\otimes \hat{\sigma}^- 
\end{align}
where $\hat{\sigma}^\pm= \frac{1}{2}(\hat{\sigma}^z\mp \ii\hat{\sigma}^y)$ is the $\mathfrak{su}(2)$ ladder operator {associated with the $\ket{\pm}$ basis}, i.e., $\hat{\sigma}^+\ket{-}=\ket{+}$ and $\hat{\sigma}^-\ket{+}=\ket{-}$. By linearity, one can work out easily the time evolution of operators of the form
\begin{align}
    A &= \begin{bmatrix}
        a_1 W(f_1) & a_2 W(f_2) \\ a_3 W(f_3) & a_4 W(f_4)
    \end{bmatrix}
\end{align}
for any $a_i\in \C$ and $f_{j}\in \CS$. One gets
\begin{align}
    \alpha_t(A) 
    &= \displaystyle
    \begin{bmatrix}
    a_1 e^{\ii\varphi(t)}W(f_{\bk}^{1} e^{\ii\omega_\bk t}) &  a_2 W(f_{\bk}^{2}e^{\ii\omega_\bk t}+2\lambda F_\bk(t))   \\ a_3 W(f_{\bk}^{3}e^{\ii\omega_\bk t}-2\lambda F_\bk(t)) 
    & a_4 e^{-\ii\varphi(t)}W(f_{\bk}^{4} e^{\ii\omega_\bk t})
     \end{bmatrix}
\end{align}
where $f_\bk^j$ is the $\bk$-space representation of $f_j$ using Eq.~\eqref{eq: f-fourier}.

\subsection{Interacting ground state and thermal state}

In what follows we assume $\Omega=0$ unless otherwise stated. First, we showed in Section~\ref{sec: UDW-model} that demanding the interaction to be compactly supported in spacetime is \textit{not} sufficient in general to guarantee that the UDW model admits a ground state in the full interacting theory. However, the results can be interpreted in several ways assuming that the interaction is sufficiently localized in spacetime --- namely, that the spacetime smearing is \textit{at least} ``strongly supported'' in $\M$: by this we mean either it is compactly supported or it has tails whose Fourier transforms decay superpolynomially. Then we have the following results:
\begin{enumerate}[label=(\arabic*),leftmargin=*]
    \item For massive scalar fields in any spacetime dimension, the UDW model has a well-defined ground state that lives in $\C^2\otimes\fock$. This is because the mass regulates the IR and the strong support of the spatial smearing regulates the UV. In other words, the interacting Hilbert space coincides with the tensor product of the free Hilbert spaces of the detector and the free scalar field. This would also occur for massless scalar fields that non-minimally couple to curvature since {the curvature term $\xi R\phi^2$ that appears in the Klein-Gordon equation \eqref{eq: KGE} with $\xi\neq \frac{n-1}{4(n-2)}$ }acts as an effective mass. 

    \item For {conformally coupled ($\xi = \frac{n-1}{4(n-2)}$ in Eq.~\eqref{eq: KGE})}
    massless scalar fields in (4+1)-dimensional spacetimes and higher ($n\geq 4$), $R_1(\lambda F_\bk,n)$ is always IR-finite even if $F_\bk\sim |\bk|^{-1/2}$ at small $|\bk|$ (e.g., for Gaussian smearing), hence the full Hamiltonian is self-adjoint and bounded from below. However, for $R_2(\lambda F_\bk,n)$ to be finite the spatial profile requires a small IR regulator, i.e., $F_\bk\sim |\bk|^\epsilon$ at small $|\bk|$ for any $\epsilon>0$. Without this requirement, the interacting ground state does not live in the free Hilbert space and Haag's theorem applies: that is, the interaction picture calculations using the free Hilbert space will not agree with the calculations using interacting Hilbert space unless some sort of ``renormalization'' is performed \cite{derezinski2004scattering}. 

    \item For {conformally coupled} massless scalar fields in (3+1) dimensional spacetimes, the Hamiltonian is bounded from below if we mildly regulate the IR to make $R_1(\lambda F_\bk,3)<\infty$ by demanding that Eq.~\eqref{eq: mild-IR-smearing} holds. To ensure finite number of soft bosons, we need $R_2(\lambda F_\bk,3)<\infty$ by demanding that Eq.~\eqref{eq: harsh-IR-smearing} holds. Crucially, it means that spacetime smearings that do not regulate the IR --- such as Gaussian switching and Gaussian smearing --- will \textit{not} work because the exponential decay of the spacetime smearing only resolve the UV part of Haag's theorem. Just as in (2), if we do not impose sufficient IR regularity, the interacting Hilbert space will not be a tensor product of the standard Fock space even if we have the mild IR regularity to make the Hamiltonian bounded from below.
\end{enumerate}
Since (1) and (2) are the ``easy'' cases, we focus on (3).

Observe that for massless scalar fields in (3+1)-dimensional spacetimes, we are essentially faced with two choices. The first is to impose that the spatial profile has sufficient IR regularity to make $R_2(\lambda F_\bk,3)<\infty$, which restricts the number of soft bosons produced by the interaction, and hence the interacting Hilbert space fits into the free theory, i.e., $\ket{\psi_0}\in \C^2\otimes\fock$. This restricts the possible choice for the spatial profile $F$ but will make interaction picture computation reliable since the Hilbert space is unitarily eqivalent to the free Hilbert space. {The other, possibly more interesting option, is to allow for infinitely many soft bosons ($R_2(\lambda F_\bk,3)=\infty$). The fact that there exists Hilbert space representations that can accommodate infinite particle number is well-known in quantum statistical mechanics and QFT \cite{bratteli2002operatorv2,pillet2006open,jakvsic1996model,morfa2012deformations}.}

\begin{assumption}[Infinite soft bosons]
    The choice of spatial smearing $F$ is such that Assumptions~\ref{assumption: sensible} and \ref{assumption: finite-energy} hold, but $R_2(\lambda F_\bk,3) = \infty$.
    \label{assumption: infinite-soft-bosons}
\end{assumption}
The question then boils down to which Hilbert space the ground state should live in when Assumption~\ref{assumption: infinite-soft-bosons} holds. There are two possible \textit{a priori} routes that one can take. The first possibility is to consider algebraic state $\omega$ where the GNS reconstruction theorem gives rise to a \textit{non-Fock representation}. This argument can be traced to Greenberger and Licht \cite{greenberg1963quantum} where they show that an interacting QFT whose connected $n$-point functions vanish after some finite order has trivial scattering, i.e., it is in fact a non-interacting theory. The significance of this is that the GNS Hilbert space of a pure quasifree state is unitarily equivalent to the vacuum Fock space (see discussions in \cite{earman2006haag}), hence in general one cannot expect the interacting Hilbert space to be built out of the Fock vacuum. One well-known non-Fock representation is the so-called \textit{coherent representation} used for constructing \textit{asymptotic Hilbert spaces} in scattering theory (see  \cite{derezinski2004scattering,derezinski2003vanHove} in the context of van Hove and Pauli-Fierz models). 

The second possibility is to consider algebraic state $\omega$ whose GNS representation gives a Fock representation that is not unitarily equivalent to the vacuum sector and supports states with infinitely many bosons {with respect to the number operator} of the vacuum GNS representation $\mathcal{H}_{\omega_0}$. Two such representations are well-known: they are the \textit{Kubo-Martin-Schwinger} (KMS) \textit{thermal representation} and the \textit{Araki-Woods representation}. These representations are also distinguished by the fact that they are both Fock representations that are unitarily inequivalent to the vacuum representation. The KMS representation is familiar in the context of Unruh effect: the expectation value of the number operator of Rindler observers (with proper acceleration $a$) with respect to the Minkowski vacuum is infinite, with number \textit{density} given by the Planckian distribution. Operator-algebraically this means that the Minkowski vacuum is an $(\alpha_\tau,\beta)$-KMS state with respect to the Rindler time $\tau$ and $\beta = 2\pi/a$. This suggests that we can define the ground state as zero-temperature limit of some $(\alpha_\tau,\beta)$-KMS state, and indeed this observation is known for a long time \cite{fannes1988equilibrium,spohn1989spinboson,amann1991spinboson}. 

We will be primarily interested in the second possibility since the non-Fock representations are mainly relevant for scattering theory. This is also natural from physical point of view since we should be able to describe closed-system thermal equilibrium between the detector and the field, i.e., we expect that $(\alpha_\tau,\beta)$-KMS state for any finite $\beta$ must exist for the UDW model. 

\begin{definition}[\cite{bratteli2002operatorv2}]
    \label{def: KMS}
    A state $\omega_{\beta}^{0}$ is an $(\alpha_t^0,\beta)$-KMS state with respect to thermal time $t$ if for all $A,B\in \mathfrak{B}''_{\alpha_t^0}$ a weakly-dense $\alpha_t^0$-invariant subalgebra of $\mathfrak{B}''$, we have
    \begin{align}
        \omega_\beta(A\alpha_{\ii\beta}^{0}(B)) = \omega_\beta(BA)\,.
    \end{align}
\end{definition}
\noindent Here $\mathfrak{B}''$ refers to the \textit{bicommutant} of $\mathfrak{B}$ which is a von Neumann algebra \cite{bratteli2013operator,bratteli2002operatorv2}\footnote{For our purposes, it is sufficient to know that the bicommutant $\mathfrak{B}''$ guarantees good functional-analytic properties of the algebra of observables as it is the closure of $\mathfrak{B}$ in both the strong and weak operator topologies.  Therefore the (local) algebra of observables in QFT is often defined after taking the bicommutant of a unital $C^*$-subalgebra of $\mathcal{B}(\mathcal{H})$. See \cite{sorce2023intuitive,sorce2023notes,witten2018entanglement} for accessible introduction to von Neumann algebra in physics and its role in the definition of thermality in infinite systems.}. The claim is that for the spin-boson model, and hence the UDW model, there exists such a joint equilibrium state.
\begin{proposition}
    Let $\hat{H}_{0,\Delta, \lambda}$ be the Hamiltonian of the UDW model with zero gap and possibly nonzero detuning. There exists a $(\alpha^0_t,\beta)$-KMS state of $\mathfrak{B}$ satisfying the appropriate regularity and continuity conditions in Section~\ref{sec: UDW-model}. Furthermore, in the $\ket{\pm}$ basis it is given by
    \begin{align}
        \omega_{\beta,\Delta}^0 &= \begin{bmatrix}
            \frac{1}{1+e^{-2\beta\Delta}}\omega_{\beta,+} & 0 \\ 0 & \frac{1}{1+e^{2\beta\Delta}}\omega_{\beta,-}
        \end{bmatrix}
    \end{align}
    where $\omega_{\beta,\pm}:\tilde{\mathcal{W}}(\R^n)\to \C$ such that
    \begin{align}
        \omega_{\beta,\pm} (W(g_\bk)) &= e^{-\frac{1}{2}(g_\bk, \coth(\frac{\beta\omega_\bk}{2})g_\bk)_{\mathcal{H}}} e^{\pm 2\ii\Im \bigr(\frac{\ii\lambda  F_\bk}{\omega_\bk},g_\bk\bigr)_{\mathcal{H}}}
    \end{align}
    for all $g_\bk$ in $\mathcal{H}$. 
\end{proposition}
\noindent The proof essentially parallels the one given in \cite{fannes1988equilibrium} so we do not repeat it here. We note that the KMS condition in the context of QFT is technically slightly different, in that it is intimately tied to the ``direction of time'' that is not unique in relativistic settings (see the discussions in \cite{verch2005distillability}). 

We make two comments on this result. First, observe that if there is no interaction ($\lambda F_\bk=0$) then $\omega_{\beta,+}=\omega_{\beta_-} \equiv \omega_{\beta,f}$ is precisely the KMS state of the \textit{free field theory}, and the joint state is the tensor product of the thermal states of each subsystem, i.e.,
\begin{align}
    \lim_{\lambda\to 0}\omega_{\beta,\Delta}^0 &\equiv  \omega_{\beta,\Delta}^{d} \otimes \omega_{\beta}^{f}\,,
\end{align}
where $\omega_{\beta,\Delta}^{d}$ is the KMS state of the detector, i.e., for all $X\in M_2(\C)$ we have
\begin{align}
    \omega_{\beta}^{d}(X) &= \tr(\rho_\beta X)\,,\quad \rho_{\beta} = \frac{e^{-\beta\Delta\hat{\sigma}_x}}{\tr (e^{-\beta\Delta\hat{\sigma}_x})}\,.
\end{align}
Second, and perhaps more importantly, the KMS state $\omega_{\beta,\Delta}^{0}$ is generically not unique because of the possible existence of \textit{Bose-Einstein condensation} (BEC), although this has been proven to not occur in $(1+1)$-dimensional case \cite{fannes1988equilibrium} (in which case $\omega_\beta^0$ is the \textit{unique} $(\alpha_t,\beta)$-KMS state). However this will not concern us since our goal is to study a ground state of the UDW model, and the one obtained in this fashion will not contain a condensate\footnote{A simple can be found in \cite{brunetti2021algebraic} in the context of free massive complex scalar field with continuous internal $U(1)$ symmetry.} (see \cite{spohn1989spinboson} for more detailed discussions). If we exclude BEC, then the KMS state above is unique. 

Observe that the limits $\Delta\to 0$ and $\beta\to \infty$ do \textit{not} commute, so one has to be careful in defining the ground state. For non-zero bias the ground state constructed this way (without condensate) is essentially unique.
\begin{proposition}[zero-temperature limit \cite{spohn1989spinboson}]
    The limit 
    \begin{align}
        \omega_{\infty,\Delta}^0 \coloneqq \lim_{\beta\to\infty}\omega_{\beta,\Delta}^0
    \end{align}
    exists and is a ground state for the UDW Hamiltonian $\hat{H}_{0,\Delta,\lambda}$. Furthermore, the limits $\Delta\to 0^\pm$ are related by the $\Z_2$-symmetry:
    \begin{align}
        \lim_{\Delta\to 0^\pm}\omega^{0}_{\infty,\Delta} &= \omega_{\infty,\pm}^0
    \end{align}
    where 
    \begin{align*}
        \omega_{\infty,+}^0 = \begin{bmatrix}
            \omega_{\infty,+} & 0 \\ 0 & 0 
        \end{bmatrix}\,,\quad \omega_{\infty,-}^0 = \begin{bmatrix}
            0 & 0 \\ 0 & \omega_{\infty,-} 
        \end{bmatrix}\,.
    \end{align*}
    
\end{proposition}
\noindent The zero-bias limit is the analog to the degenerate ground states we found earlier but is valid under Assumption~\ref{assumption: infinite-soft-bosons}, we can check by direct computation:
\begin{proposition}
    Let $\omega_{0,\pm} \equiv \braket{\psi_{0}^\pm|(\cdot)|\psi^\pm_0}$ be two degenerate ground states of the UDW model with finite $R_2(\lambda F_\bk,3)$. Then for any $g_\bk\in \mathcal{H}$ we have
    \begin{align}
        \omega_{0,\pm}(W(g_\bk)) = e^{-\frac{1}{2}\braket{g_\bk,g_\bk}}e^{\pm 2i\Im\bigr(\frac{i\lambda F_\bk}{\omega_\bk},g_\bk\bigr)_\mathcal{H}}\,.
    \end{align}
\end{proposition}
\noindent In other words, the expectation values of the Weyl generators using $\omega_{0,\pm}$ are functionally identical to those of $\omega_{\infty,\pm}^0$, except the fact that in one case $R_2(\lambda F_\bk,3)$ is finite and in another case it is infinite. This highlights the fact that a proper representation of the CCR algebra is necessary to fit different physical scenarios.

\subsection{Detector-field correlations at equilibrium}

In the standard interaction picture calculations, typically one assumes that in the asymptotic past the initial state is the uncorrelated ground states of the detector and the field. In general, a coupled system is not expected to have the same ground state as its uncoupled counterpart. We have also seen that in the case where the joint Hilbert space is $\C^2\otimes\fock$, the ground state does not coincide with the tensor product of the qubit ground state and the Fock vacuum; instead, it is a tensor product of the qubit ground state and \textit{displaced vacuum} (i.e., a coherent state). 

We now have all the tools to ask simple questions about how correlations can be produced via the interactions in the joint interacting model. As a first try, let us show that the interacting ground state associated with $\hat{H}_{0,\Delta,\lambda}$, defined as the zero-temperature limit of a thermal state, is also separable with respect to the detector-field bipartition: this follows from the fact that
\begin{align}
    \omega_{\infty,\pm|\Delta|}^0 = \ketbra{\pm}{\pm}\otimes \omega_{\infty,\pm} \equiv \omega_{\pm,d}\otimes \omega_{\infty,\pm }
\end{align}
where $\omega_{\pm,d}(A)=\tr(\ketbra{\pm}{\pm} A)$ corresponds to the detector's algebraic pure state and the field state $\omega_{\infty,\pm}$ is pure. One way to see this is to observe that formally $\omega_{\infty,\pm}$ is the vacuum state displaced by the ``dressing operator'' that has divergent coherent amplitude (because $R_2(\lambda F_\bk,3)=\infty$), so this is the analog of $\ket{\pm}\otimes \ket{\mp\alpha}$. This also follows from  the fact that the Schmidt rank of a separable state is 1. Since one subsystem is finite-dimensional (in this case the detector), the Schmidt rank of the joint system is given by the rank of the density operator $\rho_d$ associated with the detector's expectation values \cite{van2023schmidt}
\begin{align*}
    \tr(\rho_d A) \equiv \omega(A\otimes\openone)\,.
\end{align*}
The joint state is therefore separable since $\rho_d=\ketbra{\pm}{\pm}$ has Schmidt rank equal to 1.

In contrast, the joint finite-temperature state will generically be entangled across the bipartition, since the factorization above is no longer possible due to the relative phase. One exception is in the high-temperature regime $\beta\to 0$ where
\begin{align*}
    \omega_{\beta,\Delta} &\approx  \tau\otimes\omega_{\beta,\pm}
\end{align*}
where $\tau:M_2(\C)\to \C$ is a \textit{tracial state}, i.e., an algebraic state satisfying $\tau(AB)=\tau(BA)$. The tracial state has the property that $\tau(A)\equiv \frac{1}{2}\tr(A)$ for all $A\in M_2(\C)$, hence it corresponds to a maximally mixed state $\rho=\openone/2$, as we would expect from high-temperature limit of the Gibbs state. {This is one way to justify what we mean by thermal state for gapless detector model: we ask for the KMS state of the joint system and show that it is consistent with $\hat{\rho}_{d} \equiv \openone/2$, even though on its own the qubit does not have a length scale to make sense of its own temperature.} 

In the standard UDW model, the detector-field interaction can generate entanglement: that is, if the initial state is separable, the joint state will evolve to some entangled state due to the interaction between them. Clearly, this does not contradict the fact that {the interacting ground state is separable between the partition of the qubit and the field since the initial state in the standard UDW framework is not an equilibrium state. By construction the ground state is an eigenstate of the \textit{full} Hamiltonian}, hence evolving the interacting ground state will not produce any entanglement. More generally, one can check that the qubit channel $\Phi_t$ defined by
\begin{align}
    \Phi_t(\rho) &= \tr_\phi(\hat{U}\rho\otimes\ketbra{\mp\alpha}{\mp\alpha}U^\dagger)
\end{align}
where $\hat{U} = \exp(-\ii t\hat{H}_{0,\Delta,\lambda})$ is non-entangling for any coherent state $\ket{\mp\alpha}$  if $\rho$ is diagonal in the $\ket{\pm}$ basis as $\bigr[\ketbra{\pm}{\pm},\hat{H}_{0,\Delta,\lambda}\bigr]=0$. 

The analogous computation to the standard framework would be to consider initial states that are not an equilibrium state of the joint Hamiltonian. For example, since the Hamiltonian is chosen to have vanishing gap, an initial state that would evolve to some entangled state is given by
\begin{align}
    \omega_\phi &= \ketbra{g}{g}\otimes\omega_{0}\,.
\end{align}
where $\ket{g}$ is the eigenstate of $\hat{\sigma}_z$ and $\omega_0$ is some algebraic ground state. This example is relevant if we consider the gapless model as an approximation where the detector gap is much smaller than any other frequency scale in the system. For brevity let us denote corresponding algebraic state by $\omega_\infty^0$. The idea is that if we start off from initially separable observables, i.e., observables of the form
\begin{align}
    {A}\otimes{B}\in \mathfrak{B}\,,
\end{align}
then the separability of the ground state implies that
\begin{align}
    \omega_\infty^0(A\otimes B) &= \omega_{\pm,d}(A)\omega_{\infty,\pm}(B) 
\end{align}
for some $\omega_d,\omega_f$ being the (algebraic) states of the detector and the field respectively. Correlations are therefore generated by the interaction if
\begin{align}
    \omega_\infty^0(\alpha_t(A\otimes B)) \equiv (\alpha_t^*\omega_\infty^0)(A\otimes B) &\neq \omega_{d,t}(A)\omega_{f,t}(B)\notag 
\end{align}
for some state $\omega_{d,t}$ of the detector and $\omega_{\phi,t}$ of the field that depend on $t$ and $\alpha_t^*\omega_\infty^0$ is the pullback of $\omega_\infty^0$. Indeed, by direct computation we can show that the resulting algebraic state takes the form
\begin{align}
    \alpha_*^t\omega_\infty^0 &= 
    \frac{1}{2}
    \begin{bmatrix}
    \omega_{++} & \omega_{+-} \\ \omega_{-+}& \omega_{--}     
    \end{bmatrix}
\end{align}
where for any $A\in \mathcal{W}(\M)$ we have
\begin{align}
    \omega_{ab}(A) = \omega_0(\hat{D}(a\alpha_\bk)^\dagger A\hat{D}(b\alpha_\bk))\,,\quad a,b=\pm\,.
\end{align}
{It can be checked that this state is a pure state with Schmidt rank 2, hence it is a genuine entangled state and the von Neumann entropy of the reduced states of either subsystem can be computed exactly \cite{tjoa2023nonperturbative}.}

{In summary, the procedures for studying how interactions generate correlations at the level of the states or the observables are the same as the standard interaction picture scenario. The technical advantage of using the operator-algebraic approach is that we are able to ensure that there is a well-defined ground state for the UDW model to make sense and we can find appropriate Hilbert space representations to support them. The time evolution can then be implemented using the evolution map $\alpha_t$. Note that the classical trajectory of the qubit is essentially hidden in the spacetime smearing function and in principle we can choose any trajectory and perform quantization with respect to any choice of global time but only a few special cases are amenable to simple explicit computations. }

\section{Discussion and outlook}
\label{sec: discussion}

\subsection{Summary}

Let us summarize the main points in this work. In quantum statistical mechanics and relativistic QFT, it is well-known that there exists (infinitely many) unitarily inequivalent representations of the CCR algebra for bosonic fields (even for the non-interacting case), and an analogous statement holds for fermionic fields with anti-commutation relations. The corresponding Hilbert spaces associated to two inequivalent representations are disjoint. The van Hove model serves as a textbook example of an ``interacting'' theory where the vacuum representation of the model is generically unitarily inequivalent to the vacuum representation of the free bosonic field with zero current, thus demonstrating the relevance of Haag's theorem even in the simplest Gaussian settings \cite{fewster2019algebraic}. In particular, the interaction picture, which relies on the zero-current Hilbert space representation, is of limited applicability unless one imposes sufficient constraints on the UV and the IR behaviour of the model. 

Our work is motivated by the observation that the UDW model is a simple variant of the SB model \cite{fannes1988equilibrium,spohn1989spinboson,amann1991spinboson,hasler2011ground,hasler2021existence}, where the choice of coupling function is fundamentally determined from the spatial profile of the interaction in the background spacetime. By construction, such UDW model is an interacting theory of a qubit and a scalar field that has good UV behaviour due to spatial localization of the detector-field interaction, thus any issues related to the inequivalent representations are largely due to the IR behaviour of the interaction. By using the zero-gap limit of the UDW Hamiltonian, we can work out the constraints on the spatial smearing function the detector so that the interacting Hilbert space ``fit'' into the non-interacting Hilbert space (the tensor product of qubit Hilbert space and the zero-current bosonic field). By turning the issue of inequivalent representations around, our work provides us with some heuristics under which the UDW model calculations in the interaction picture, which depend on non-interacting Hilbert space representations, are reliable. For example, in Section~\ref{sec: consequences} we list some examples where this complication does not arise, such as for massive scalar fields and for massless scalar fields that couple to curvature (non-conformal coupling) and gain effective mass. 

Our work requires that we view the UDW model in the standard paradigm of {closed} interacting systems, where the joint system has a time-independent Hamiltonian. Such a paradigm is used from the Rabi model in quantum optics to the Standard Model of particle physics. On the other end, time-dependent Hamiltonians imply the existence of classical external control that mediates the interaction Hamiltonian. Such a paradigm is used in quenched quantum systems or Floquet systems (time-periodic driven systems). The standard UDW Hamiltonian is closer to the latter where the switching function serves as an external control, however rigorous results on time-dependent unbounded Hamiltonians for infinite systems is much more lacking than the time-independent counterpart. Our argument is based on reframing the UDW Hamiltonian within the time-independent framework and then adapt the rigorous results from the SB model.

Overall, we have shown that by lifting the rigorous results on the spin-boson model to the UDW model with suitable modifications, the IR regularity of the model  {contains} essentially the same  {physical questions as those posed in} Haag’s theorem for the van Hove model. On the one hand, it means that the interaction picture should be used judiciously in the UDW framework. On the other hand, we can turn this around and say that for many interesting applications, our work tells us how to avoid the obstruction altogether and instead work in the settings where the non-interacting Hilbert space can be used. Furthermore, our work suggests that conceptually, we should separate the localization of interactions from the localization of observables: this has implications for further analysis involving multiple detectors, e.g., in the context of measurement theory of QFTs, relativistic causality, field-mediated communication, and entanglement harvesting protocols. Last but not least, our work shows that in a way the IR  {behaviour} is familiar to us: in high energy physics this is essentially a toy version of the soft photon problem, and in condensed matter this is a toy version of the Bose-Einstein condensation.  

We should mention two technical points. Firstly, since we are using the time-independent reformulation of the UDW model, in principle this does \textit{not} prove that whenever the (time-independent) UDW model has ``good'' IR behaviour then the standard (time-dependent) UDW model will also inherit this property. By ``good'' we mean that the Hilbert space representation of the UDW model will be unitarily equivalent to the non-interacting theory (i.e., when the detector is not coupled to the field). Since at a fixed time the two Hamiltonians are essentially identical, our argument shows that whenever the UDW model has ``bad'' IR behaviour, the standard time-dependent UDW model should not be expected to have good IR behaviour. For example, we argued that in some cases the well-used Gaussian smearing function in flat spacetime will not suffice. Secondly, even in the original SB model, one does not quite have rigorous results for all choice of parameters. In particular, due to the possibility of Bose-Einstein condensation and phase transition \cite{bratteli2002operatorv2}, the existence of ground states of the SB model (hence also the UDW model) can really be proven only for zero-gap limit \cite{fannes1988equilibrium} or nonzero gap but small detector-field coupling \cite{spohn1989spinboson,amann1991spinboson,hasler2011ground,hasler2021existence}.

\subsection{Further directions}

Our work poses some open questions that are closely related to the existing RQI literature on the use of UDW model in various contexts. We list some of them below.

The first natural extension of our work is for multiple detectors\footnote{Note that the operator-algebraic formulation for a single UDW detector has actually been used once to study the Unruh effect \cite{deBievre2006unruh}. They also use the time-independent formulation but the focus was on the thermalization of an accelerating qubit detector.}. To the best of our knowledge, this has not been worked out even in the original SB model. Presumably in the SB model, having multiple qubits do not modify the problem in important ways, but in the UDW model the underlying spacetime has a causal structure. Thus the choice of the coupling functions are intimately connected to spacetime localization of the detectors. This problem requires that the notion of \textit{spacelike-separated} qubits be modified, possibly using the clustering-type argument in QFT, since the total Hamiltonian in this framework will not temporally localized. This affects the interpretation of certain tasks such as the entanglement harvesting protocol \cite{pozas2015harvesting,pozas2016entanglement,reznik2003entanglement,reznik2005violating,Valentini1991nonlocalcorr}, which is sensitive to the causal relationships between the qubit interaction regions \cite{tjoa2021harvesting,Caribe2023harvestingBH}. One also would like to understand what happens when only one detector couples with a ``bad'' IR behaviour: could the other detector see it, or would this only affect the choice of Hilbert space representations one works in (e.g., by choosing Araki-Woods representation)? 

Second, it is worth noting that the potential use of operator-algebraic techniques for two-qubit entanglement harvesting has already been somewhat anticipated in \cite{verch2005distillability}. The idea is to rigorously justify the pioneering work on entanglement harvesting protocol \cite{reznik2003entanglement,reznik2005violating,Valentini1991nonlocalcorr} and make connections to the rigorous understanding of entanglement in quantum field theory. Essentially, the natural question is this: how much can we say about the relationships between the distillability, positive partial transpose property, and entanglement measures between the two subsystems of the field that couple to the two detectors? A precise connection between them would provide a more (if not the most) complete answer to the question of how much information can finite-dimensional quantum-mechanical probes obtain from infinite-dimensional quantum field through physically motivated couplings. We find it rather fortunate that there has been quite a lot of progress in our understanding of entanglement theory in infinite dimensions (see, e.g., \cite{vanLuijk2023schmidt,vanLuijk2024embezzlement,hollands2017entanglement,sanders2023separable}). 

Third, for simplicity we have not dealt with the scenario where the detector has nonzero gap $\Omega$ in addition to (possibly) nonzero detuning/bias parameter $\Delta$. In \cite{fannes1988equilibrium}, the approach to this is to treat small gap as a bounded perturbation to the gapless model we consider here; in \cite{hasler2021existence}, the strategy is to instead treat the non-interacting spin and bosonic field as the primary Hamiltonian and then consider unbounded perturbation through the spin-boson interaction. These problems are strongly connected to phase transition and symmetry-breaking: it was shown in \cite{spohn1989spinboson} for gapped qubit, there is critical coupling strength beyond which the joint ground state ``leaves the free Hilbert space'' when the bosonic field has Ohmic spectral density. One can now ask: in the context of QFT in curved spacetimes, if we view the UDW-type interaction as being physically relevant for probing the Unruh and Hawking effects, are these affected by the possibility of having Bose-Einstein condensation and the generation of infinitely many IR bosons due to the detector-field interaction? 

Last but not least, it is noteworthy that many aspects of QFT, such as having local algebras that are not Type I von Neumann algebras, also arise in infinite quantum spin chains. Since spin chains already can exhibit various types of phase transitions, it is of interest to see whether the kind of interactions considered here also carries through to spin systems. The point is to identify which part of the physics has to do intimately with relativity and which part has to do ``merely'' with the fact that there are infinitely many degrees of freedom in the system of interest.

We leave these questions for the near future.

\section*{Acknowledgment}

E. T. also thanks Robert Jonsson for hospitality at Nordita, Sweden during which part of this work is performed, and also Roie Dann for useful discussions. E. T. acknowledges funding from the Munich Center for Quantum Science and Technology (MCQST), funded by the Deutsche Forschungsgemeinschaft (DFG) under Germany’s Excellence Strategy (EXC2111 - 390814868).

\appendix 

\section{Convention for propagators}
\label{appendix: green-function}

Here we briefly review the conventions used for defining the causal propagator in quantum field theory and some of the ambiguities that may arise, since it caused quite a bit of headache for the authors.  The issue has to do with the fact that the causal propagator $E$ is a \textit{kernel} of the wave equation, in the sense that for the Klein-Gordon wave operator $\hat P$ we have
\begin{align}
    \hat PE = 0
\end{align}
in the distributional sense. In particular, it means that $\pm E$ is in $\ker \hat P$. 

Our convention, which follows \cite{tjoa2022holographic-essay}, can be regarded as the ``least-minus-sign convention'' in that our definition removes as many minus signs as possible in the formalism from any other conventions. For concreteness, we compare our convention with that of \cite{KayWald1991theorems} where the metric signature is (fortunately) positive and is the closest to our convention. The CCR of the field operator is $[\hat\phi(f),\hat{\phi}(g)]=-\ii E(f,g) = +\ii \sigma(Ef,Eg)$. This occurs because for the definition of the field operator, they use the symplectic form $\tilde\sigma$ on $\Sol_\C(\M)$ that has \textit{future-directed} unit normal, i.e., $\sigma=-\Tilde{\sigma}$. In other words, the smeared field operator reads
\begin{align}
    \hat{\phi}(f)\coloneqq \sigma(Ef,\hat{\phi}) \equiv -\sigma(\hat{\phi},Ef)= \tilde{\sigma}(\hat{\phi},Ef) \,.
\end{align}
Our convention imposes that $\sigma(Ef,Eg) = E(f,g)$ and removes the minus sign difference by having the unit timelike normal 1-form of a spacelike Cauchy surface be \textit{past-directed} \cite{poisson2009toolkit}. Notice that the Weyl generators in \cite{KayWald1991theorems} are also defined with a minus sign:
\begin{align}
    W_{\textsc{kw}}(Ef)\coloneqq e^{-\ii\hat{\phi}(f)}
\end{align}
which gives the Weyl relation
\begin{align}
    W_{\textsc{kw}}(Ef)W_{\textsc{kw}}(Eg) = e^{-\frac{\ii\tilde{\sigma}(Ef,Eg)}{2}}W_{\textsc{kw}}(Ef+Eg)\,.
\end{align}
In any case, it is a convention so one is free to choose which one to adopt.

The ambiguity is not yet fully resolved from the symplectic structure alone, because in \cite{KayWald1991theorems} they define $E$ also as advanced-minus-retarded propagator, so in order to be consistent with our convention it must be the case that somehow
\begin{align}
    E^{\text{R/A}}_{\textsc{kw}} = -E^{\text{R/A}}\,.
    \label{eq: kw-green}
\end{align}
The minus sign does not change the spacetime supports of these bi-distributions, so it will not affect the formalism we use in this work. However, it does affect the convention for the Green's function technique for solving the wave equation. In the literature of classical wave theory, the standard convention in flat spacetime is \cite{poisson2011motion}
\begin{align}
    (\partial_t^2-\nabla^2)G(\sx,\sx') &= \delta^{n+1}(\sx-\sx')\,.
\end{align}
For mostly-plus metric signature, this is typically written as
\begin{align}
    \hat P G(\sx,\sx') = -\delta^{n+1}(\sx-\sx')\,,
\end{align}
while in the QFT literature, the convention is to instead write it as \mbox{$(-\hat{P})G(\sx,\sx') = \delta^{n+1}(\sx-\sx')$}. Following \cite{poisson2011motion}, the retarded and advanced propagators are given by 
\begin{equation}
    \begin{aligned}
    E^{\text{R}}(\sx,\sx') &\coloneqq +\frac{1}{4\pi|\bx-\bx'|}\delta(t-t'-|\bx-\bx'|)\,,\\
    E^{\text{A}}(\sx,\sx') &\coloneqq +\frac{1}{4\pi|\bx-\bx'|}\delta(t-t'+|\bx-\bx'|)\,,
    \label{eq: adv-ret-green}
    \end{aligned}
\end{equation}
{where the supports of the solution $E^{\text{R/A}} f$ is contained in the causal future and causal past of the support of $f$, i.e., }
\begin{align}
    \supp (E^{\text{R/A}} f)\subseteq J^\pm (\supp f)\,.
\end{align}
Comparing this with the vacuum expectation value of the field commutator (as done in \cite{tjoa2021harvesting,Causality2015Eduardo}), we see that it is consistent with
\begin{align}
    [\hat{\phi}(\sx),\hat\phi(\sx')]=\ii (E^\text{A}-E^\text{R})(\sx,\sx')\openone
\end{align}
and hence our definition of causal propagator is indeed the advanced-minus-retarded propagator. Importantly, this means that our convention for the Green's function (which follows that of \cite{poisson2011motion}) is the one where
\begin{align}
    \hat PE^{\text{R/A}} f = -f\,.
\end{align}
Alternatively, we could interpret our $E^{\text{R/A}}$ as distributional inverse of $-\hat{P} \equiv -\nabla_a\nabla^a +m^2 + \xi R$: this is used fairly often especially when the metric signature is mostly-minus.

Finally, comparing this to the convention by \cite{KayWald1991theorems}, we have the relationship
\begin{align}
    \hat PE^{\text{R/A}} f = -f \Leftrightarrow  (-\hat P)E^{\text{R/A}} f = f \Leftrightarrow \hat PE^{\text{R/A}}_{\textsc{kw}}f = f  
    \label{eq: wave-inverse}
\end{align}
Therefore, the convention in \cite{KayWald1991theorems} is such that $E^{\text{R/A}}_\textsc{kw}$ is a distributional inverse of  $P$ (mostly-plus signature). Indeed, since their commutator reads
\begin{align}
    [\hat{\phi}(\sx),\hat\phi(\sx')]= -\ii E_\textsc{kw}(\sx,\sx')
\end{align}
it is necessary that they define 
\begin{align}
    {
    E^{\text{R/A}}_\textsc{kw}(\sx,\sx')} &\coloneqq -\frac{1}{4\pi|\bx-\bx'|}\delta(t-t'\mp|\bx-\bx'|)\notag\\
    &\equiv -E^{\text{R/A}}
\end{align}
in order for $E_\textsc{kw}$ to be  defined as advanced-minus-retarded propagator. Alternatively, one can define instead {\textit{retarded-minus-advanced propagator}}
\begin{align}
    \tilde{E}(\sx,\sx') \coloneqq E_\textsc{kw}^\text{R}-E_\textsc{kw}^\text{A}
\end{align}
in which case $[\hat{\phi}(\sx),\hat\phi(\sx')]=\ii \Tilde{E}(\sx,\sx')$, so it takes the same functional form as our convention despite the opposing sign of Green's functions.

\section{Fock representation}
\label{appendix: Fock-quantization}

In this section we review some basic properties of Fock representation, as this will be relevant for discussing different unitarily inequivalent Fock representations that arise in the context of the spin-boson and UDW models. We benefit mainly from the combination of \cite{derezinski2006introduction,morfa2012deformations} together with standard references \cite{fewster2019algebraic,bratteli2002operatorv2}.
\begin{definition}
    Let $\mathcal{H}^{\odot n}$ be the $n$ symmetric tensor product of a one-particle Hilbert space $\mathcal{H}$. The \textbf{bosonic Fock space} with respect to $\mathcal{H}$ is defined as
    \begin{align}
        \Gamma(\mathcal{H}) \coloneqq \bigoplus_{n=0}^\infty \mathcal{H}^{\odot n}\,.
    \end{align}
\end{definition}
\noindent For $n=0$ --- the \textit{vacuum sector} --- we have $\mathcal{H}^{\odot 0}\cong \C$, and each $n$ labels disjoint subspaces called the \textit{$n$-particle sectors}. 
\noindent In the physics literature, we often see the notation $\fock$ or its variants such as $\mathcal{F}(\mathcal{H})$. In this Appendix we follow the mathematical convention to make comparisons easier, but we will use $\Gamma(\mathcal{H})$ and $\fock$ interchangeably.  

If $\hat{h}$ is a closed operator on $\mathcal{H}$, then we define two closed operators
\begin{align}
    \Gamma^n(\hat{h}) &\coloneqq \hat{h}^{\otimes n} \,,\\
    \dd\Gamma^n(\hat{h}) &\coloneqq\sum_{j=1}^n \openone^{\otimes (j-1)}\otimes\hat{h}\otimes\openone^{\otimes (n-j)}\,,
\end{align}
and hence the operators
\begin{align}
    \Gamma(\hat{h})\coloneqq\bigoplus_{n=0}^\infty\Gamma^n(\hat{h})\,,\quad 
    \dd \Gamma(h) &\coloneqq \bigoplus_{n=0}^\infty \dd\Gamma^n(h)\,.
\end{align}
are automorphisms on $\Gamma(\mathcal{H})$. These operators are self-adjoint if $\hat{h}$ is self-adjoint.

The notations are well-chosen since they are reminiscent of those found in Lie group theory. Consider a Lie group $G$, its corresponding Lie algebra $\mathfrak{g}$, together with their representations $\Pi: G\to GL(V)$ and $\pi: \mathfrak{g}\to \mathfrak{gl}(V)$. Recall that the connected component of $G$ satisfies
\begin{align}
    \pi(g) &= \frac{\dd}{\dd t}\Bigr|_{t=0}\Pi(e^{\ii t g}) \equiv \frac{\dd}{\dd t}\Bigr|_{t=0} e^{i t\pi(g)}\,.
\end{align}
where $g\in\mathfrak{g}$ and $e^{\ii tg}\in G$. In other words, the Lie algebra can be viewed as the tangent space of the corresponding Lie group and its elements are obtained as derivatives of the Lie group elements. Something similar appears in the construction of Fock spaces. 
\begin{proposition}[\cite{derezinski2006introduction}]
    Let $\hat{h},\hat{h}_1,\hat{h}_2$ be some closed operators acting on the one-particle Hilbert space $\mathcal{H}$. The following identities hold:
    \begin{itemize}
        \item $\Gamma(\hat{h}_1)\Gamma(\hat{h}_2) = \Gamma(\hat{h}_1\hat{h}_2)$
        \item $[\dd\Gamma(\hat{h}_1),\dd\Gamma(\hat{h}_2)] = \dd\Gamma([\hat{h}_1,\hat{h}_2])$
        \item $\Gamma(e^{\ii t\hat{h}}) = e^{\ii t\,\dd\Gamma(\hat{h})}$.

    \end{itemize}
\end{proposition}
\noindent We see that $\Gamma,\dd\Gamma$ takes the role of $\Pi$ and $\pi$ in the Lie group theory and they behave like representations for the algebra of operators acting on $\mathcal{H}$ by setting $V=\Gamma(\mathcal{H})$. Indeed, if $\hat{h}$ is the one-particle Hamiltonian, then
\begin{align}
    \hat{H}_0\coloneqq \dd\Gamma(\hat{h}) \equiv \int\dd^n\bk\,\omega^{\phantom{\dagger}}_\bk\hat{a}^\dagger_\bk\hat{a}_\bk^{\phantom{\dagger}}
\end{align}
is precisely the full Hamiltonian in the Fock representation\footnote{This notation is technically not correct since $\hat{a}_\bk$ is not a proper operator acting on the one-particle Hilbert space.}. Sometimes one writes $\hat{H}_0 = \dd \Gamma(\omega) $. Schematically speaking, the eigenvectors of $e^{\ii t \hat{h}}$ are $\{e^{\ii t\omega_\bk}\}$, so the eigenvalues of $\hat{h}$ are precisely the derivative of $e^{\ii t\omega_\bk}$ evaluated at $t=0$, namely $\omega_\bk$. From these, we see that the number operator would correspond to an operator with eigenvalue $n$:
\begin{definition}
    The \textbf{number operator} is defined by
    \begin{align}
        \hat{N}\coloneqq \dd\Gamma(\openone )\,.
        \label{eq: number-operator-rigorous}
    \end{align}
\end{definition}
\noindent This is the reason why the free Hamiltonian is sometimes written as $\dd \Gamma(\omega)$. 

A vector state $\Psi\in\fock$ is defined to be \textit{a sequence} $(\Psi_n)_{n\in \mathbb{N}_0}$:
\begin{align}
    {\Psi}\coloneqq (\Psi_0,\Psi_1,\Psi_2,...)
\end{align}
where each $\Psi_n\in \mathcal{H}^{\odot n}$. A vacuum state is written as $\Omega\coloneqq (1,0,0,...)$. A more physicist or elementary linear-algebraic notation would be to regard $\mathcal{H}^{\odot n}$ as subspaces of $\fock$ and write $\ket{\Psi}$ as a unique decomposition
\begin{align}
    \ket{\Psi} = \ket{\Psi_0} + \ket{\Psi_1} + \ket{\Psi_2} + ... 
\end{align}
where each $\ket{\Psi_n}\in \mathcal{H}^{\odot n}$.

The ladder operators $\hat{a},\hat{a}^\dagger$ are viewed as linear maps on $\fock$. First, let us consider the annihilation operator  defined to through the action
\begin{align}
    \hat{a}(\Phi) \ket{\Psi_{n+1}} = \sqrt{n+1} \hat{l}_{n+1}[\Phi]\ket{\Psi_{n+1}}\in \mathcal{H}^{\odot n}\,,
\end{align}
where the map $\hat{l}_{n+1}[\Phi]:\mathcal{H}^{\otimes n+1}\to \mathcal{H}^{\otimes n}$ is defined via
\begin{align}
    \hat{l}_{n+1}[\Phi](\Psi_1\otimes ...\otimes\Psi_{n+1}) = \braket{\Phi|\Psi_1}\Psi_2\otimes...\Psi_{n+1}\,.
\end{align}
and this map restricts to the map $\mathcal{H}^{\odot n+1}\to \mathcal{H}^{\odot n}$. Similarly, the creation operator is defined via \cite{fewster2019algebraic}
\begin{align}
    \hat{a}(\Psi)\ket{\Psi_n} = \sqrt{n+1}S_{n+1}[\ket{\Psi}\otimes\ket{\Psi_n}]\,,
\end{align}
where $S_{n+1}$ is an orthogonal projection onto $\mathcal{H}^{\odot n+1}$.

Morally speaking, these definitions account for the fact that both the ladder operator and the state are \textit{smeared} in the occupation number basis. Let us demonstrate this by considering the action of the smeared ladder operators on the one-particle state
\begin{align}
    \ket{\Psi_1} &\equiv \int\dd^n\bk\,\Psi(\bk)\hat{a}_\bk^\dagger \ket{\Omega} \equiv \hat{a}^\dagger(\Psi)\ket{\Omega}\,,
\end{align}
noting that the subscript ``$1$'' labels the sector (not the smearing function). Since the ladder operators at fixed $\bk$ are ill-defined, we need to consider smeared ladder operators that can be formally written as
\begin{equation}
    \begin{aligned}
            \hat{a}(\Phi^*) = \int\dd^n\bk\,\Phi(\bk)\hat{a}_\bk\,, \quad \hat{a}^\dagger(\Phi) = \int\dd^n\bk\,\Phi(\bk)\hat{a}^\dagger _\bk\,.
    \end{aligned}
\end{equation}
For the annihilation operator have
\begin{align}
    \hat{a}(\Phi^*)\ket{\Psi_1} &\equiv \hat{a}(\Phi^*)\hat{a}^\dagger(\Psi)\ket{\Omega} \notag\\
    &= \Bigr([\hat{a}(\Phi^*),\hat{a}^\dagger(\Psi)] + \hat{a}^\dagger(\Psi)\hat{a}(\Phi^*)\Bigr)\ket{\Omega} \notag\\
    &= \braket{\Phi,\Psi}_{\kg}\ket{\Omega}\,.
\end{align}
For the creation operator we have
\begin{align}
    \hat{a}^\dagger(\Phi)\ket{\Psi_1} &\equiv \hat{a}^\dagger(\Phi)\hat{a}^\dagger(\Psi) \ket{\Omega} \notag\\
    &\equiv \int\dd^n\bk\,\dd^n\bk'\,\Phi(\bk)\Psi(\bk')\hat{a}^\dagger_\bk\hat{a}_{\bk'}^\dagger\ket{\Omega}\,.
\end{align}
We have two cases. First, if $\Phi = \Psi$ then we are adding ``two particles on the same mode'' to the vacuum and we should get $\sqrt{2}(\hat{a}^\dagger(\Psi))^2\ket{\Omega}$; conversely, we should not have an additional factor of $\sqrt{2}$ if we add two particles into two different modes, which corresponds to $\Phi\not\propto\Psi$. Indeed, the algebraic definition enforces this by symmetrization through $S_{n+1}$:
\begin{align}
    S_2[\Phi\otimes\Psi] \equiv \frac{1}{\sqrt{2}}(\ket{\Phi}\otimes\ket{\Psi} + \ket{\Psi}\otimes\ket{\Phi})
\end{align}
and hence
\begin{align}
    \hat{a}^\dagger(\Phi)\hat{a}^\dagger(\Psi) = \begin{cases}
        \sqrt{2}(\hat{a}^\dagger(\Psi))^2\ket{\Omega}\qquad \Psi\propto \Phi\\
        \hat{a}^\dagger(\Phi)\hat{a}^\dagger(\Psi)\ket{\Omega}\qquad \Phi\not\propto \Psi\,.
    \end{cases}
\end{align}
which is indeed what we expect from the unsmeared version of the creation operators. These definitions have the advantage that $\hat{a}(\Phi)\ket{\Psi_1}$ will be a proper vector in $\fock$. This is to be contrasted with the physicists' liberal usage of $\hat{a}_\bk,\hat{a}_\bk^\dagger$ where rigorously speaking $\hat{a}_\bk^\dagger\ket{\Omega}$ is not normalizable and hence lives outside the Fock space ($\hat{a}_\bk$ has no densely defined adjoint \cite{fewster2019algebraic}).

The number operator $\hat{N}$ is the operator that satisfies 
\begin{align}
    \hat{N}\ket{\Psi_n} = n\ket{\Psi_n}
\end{align}
with domain being all vectors in $\fock$ with finite \textit{variance}: $\braket{\Psi|\hat{N}^2|\Psi} = \sum_{n=0}^\infty n^2||\Psi_n||^2<\infty$. The eigenvalue equation above is consistent with Definition~\ref{eq: number-operator-rigorous}. In terms of the smeared language, we first define \cite{bratteli2002operatorv2}
\begin{align}
    \hat{N}_{\Psi}\coloneqq\hat{a}^\dagger(\Psi)\hat{a}(\Psi)
\end{align}
so that formally
\begin{align}
    \hat{N} = \sum_{n=0}^\infty\hat{N}_{\Psi_n}\,.
\end{align}
To define this properly one needs to take care of domain issues since it is unbounded (see \cite{bratteli2002operatorv2} for more details). To see that this is consistent with physics literature, observe that in the sharp momentum limit we have
\begin{align}
    \hat{N}_{u_\bk} &= \hat{a}^\dagger(u_\bk)\hat{a}(u^*_\bk) 
    \equiv \hat{a}_\bk^\dagger\hat{a}_\bk^{\phantom{\dagger}} \eqqcolon \hat{N}_\bk\,.
\end{align}

We make a passing remark that the definition of the ladder operators is \textit{basis-independent} \cite{fewster2019algebraic}. To see this, suppose that we have a function $\Psi\in\mathcal{H}$, we can write
\begin{align}
    \hat{a}^\dagger(\Psi) &\equiv \int \dd^n\bk\, (u_\bk^*,\Psi)_\mathcal{H}\hat{a}(u_\bk) \notag\\
    &\equiv \int \dd^n\bk\, (v_\bk^*,\Psi)_\mathcal{H}\hat{a}(v_\bk) 
\end{align}
for two different mode functions $\{u_\bk\}$ and $\{v_\bk\}$ that are related by some Bogoliubov transformations. Indeed, in the standard canonical quantization approach (see, e.g., \cite{birrell1984quantum}) $\hat{a}^\dagger_\bk\equiv \hat{a}^\dagger(u_\bk)$ could correspond to the creation operator associated with an inertial observer (hence defines a Minkowski vacuum), while $\hat{b}^\dagger_\bk\equiv \hat{a}^\dagger(v_\bk)$ is associated to an accelerating observer. The algebraic approach is cleaner in the sense that $\hat{a}(\Psi)$ is well-defined without appealing to improper (unnormalizable) eigenmodes.

\section{Finite-density quasifree representations}
\label{appendix: araki-woods}

Here we briefly review two different quasifree representations known as the \textit{Kubo-Martin-Schwinger representation} and \textit{Araki-Woods representation} \cite{araki1963representations,pillet2006open,morfa2012deformations,derezinski2006introduction}. The important property of the both representations is that since they arise from quasifree states, the GNS construction gives rise to Fock representations that are generically not unitarily equivalent to the vacuum representation (which we recall is also a Fock representation).  For clarity it is actually easier to explain the Araki-Woods (AW) representation first since it is easier to see how the three representations are related to one another.

Let $\mathcal{H}_0$ be the one-particle Hilbert space associated with the vacuum representation, so that the Fock space of  the vacuum representation is $\Gamma(\mathcal{H}_0)$. The AW representation can be viewed as a Fock representation with Hilbert space
\begin{align}
    \mathcal{H}_{\textsc{aw}} &= \Gamma(\mathcal{H}_0)\otimes \bar{\Gamma(\mathcal{H}_0)}
\end{align}
where $\bar{\Gamma(\mathcal{H}_0)}$ is the complex conjugate of the vacuum Fock space. Using the momentum-space labeling, the ladder operators in the AW representation read 
\begin{align}
    \hat{a}_{\textsc{aw}}(\varphi_\bk) \equiv \hat{a}(
    \sqrt{1+\rho_\bk}\varphi_\bk) \otimes \bar{\openone} + \openone\otimes\hat{a}(\sqrt{\rho_\bk}\varphi_\bk) 
\end{align}
where $\varphi_\bk\in L^2(\R^n)$ with appropriate on-shell relativistic measure $\dd^n\bk/\sqrt{2(2\pi)^n\omega_\bk}$. 
\begin{proposition}[\cite{araki1963representations}]
    The AW representation is a quasifree and hence regular representation, with the expectation value of the Weyl generator is given by
    \begin{align}
        \omega_{\mathrm{AW}}(W(f)) 
        &= e^{-\frac{1}{2}(f_\bk,(1+2\rho_\bk)f_\bk)_{\mathcal{H}}}
    \end{align}
    Furthermore, the AW representation is unitarily equivalent to the KMS representation by setting the momentum density to be
    \begin{align}
        \rho_\bk(\beta) \coloneqq \frac{1}{e^{\beta\omega_\bk}-1}\,.
        \label{eq: Planck-distribution}
    \end{align}
    and $\beta$ is the inverse KMS temperature.
\end{proposition}

The KMS representation, also known as \textit{thermal representation}, is typically defined in terms of analytic properties of the two-point correlators of the field \cite{kubo1957statistical,martinSchwinger1959theory,bratteli2002operatorv2} (see Definition~\ref{def: KMS}). The KMS representation can be viewed as a Fock representation with the
 Hilbert space\footnote{The KMS representation can also be thought of as specifying a different one-particle structure $(\mathcal{H}_0\oplus\bar{\mathcal{H}_0},K_\beta)$ from the one in the vacuum representation $(\mathcal{H}_0,K)$ \cite{fewster2019algebraic,KayWald1991theorems}. Roughly speaking, the one-particle Hilbert space of the KMS representation is the ``purified'' (doubled) version of the vacuum representation.}
\begin{align}
    \mathcal{H}_{\beta} &= \Gamma(\mathcal{H}_0 \oplus \bar{\mathcal{H}_0})
\end{align}
where $\bar{\mathcal{H}_0}$ is the complex conjugate of the one-particle Hilbert space. The ladder operators in the KMS representation is related to the AW one by setting the momentum density to the Planck distribution \eqref{eq: Planck-distribution}. 
{For completeness, in terms of the Weyl generators and the cyclic vector $\ket{\Omega}$, the vacuum and the AW representations are related via \cite{morfa2012deformations}
\begin{align}
    \pi_{\textsc{aw}}(W(f_\bk)) 
    &= 
    \pi_{\omega_0}(W(\sqrt{1+\rho_\bk}f_\bk))
    \otimes 
    \bar{\pi_{\omega_0}(W(\sqrt{\rho_\bk}f_\bk))} \,,
\end{align}
where the RHS are tensor products of the vacuum representation of two Weyl generators that depend on $\rho_\bk$. The GNS cyclic vector of the AW representation is given in terms of the one in the vacuum representation as $\ket{\Omega}_{\textsc{aw}} 
= \ket{\Omega}\otimes\ket{\bar{\Omega}}$.}

For scalar field theory, a result of Takesaki (1969) \cite{takesaki1970disjointness} implies that two KMS states with different $\beta$ lead to unitarily inequivalent thermal representations. Since a vacuum state can be viewed as the zero-temperature limit of thermal states, it follows that the vacuum representation (the `default' Fock space of the scalar field theory where the vacuum state lives) is not unitarily equivalent to thermal representations or the AW representations unless the momentum density $\rho_\bk$ vanishes (or equivalently, when $\beta\to\infty$). 

Recall that when the UDW coupling satisfies $R_2(\lambda f_\bk,3)=\infty$, the physical interpretation is that there are infinite number of IR bosons produced by the interactions. Since the standard Fock representation cannot support infinitely many \textit{total} particle number, the appropriate Hilbert space to find the ground state is the AW representation. Indeed, for the Planckian distribution \eqref{eq: Planck-distribution} the total number \textit{density} (integrated over all $\bk$) is finite, but in the thermodynamic limit the total number of excitations is divergent due to the infinite spatial volume of the spacetime (unless we specifically choose compact spatial geometries which exclude Minkowski spacetime). The choice of zero-temperature limit of equilibrium KMS state as a candidate for the ground state of the UDW model is therefore an immediate consequence of the AW representation being able to support infinite particle number for finite particle number density. The KMS state, in particular, is natural since both the vacuum and the thermal states are stationary states of the Hamiltonian.

\bibliography{ref}
\bibliographystyle{iopart-num}

\end{document}